\documentclass[subscriptcorrection,upint,varvw,barcolor=Goldenrod3,mathalfa=cal=euler,balance,hyphenate,french]{asmejour} %

\hypersetup{%
	pdfauthor={Ankit Goel},                       		   	
	pdftitle={ASME Journal Paper LaTeX Template},                  	
	pdfkeywords={ASME journal paper, LaTeX template, BibTeX style, asmejour class},
	pdfsubject = {Describes the asmejour LaTeX template},			
	pdflicenseurl={https://ctan.org/pkg/asmejour},
}

\JourName{Dynamic Systems, Measurement, and Control}

\usepackage{amsmath} 

\usepackage{amsthm}

\ifdefined\begintheorem\else
    \newtheorem{theorem}{Theorem}
\fi

\newtheorem{proposition}{Proposition}[section]
\newtheorem{definition}{Definition}[section]
\newtheorem{fact}{Fact}[section]
\newtheorem{remark}{Remark}
\newtheorem{lemma}{Lemma}[section]

\usepackage{float} 

\usepackage{tikz}
\usetikzlibrary{shapes,arrows,calc,positioning}
\tikzstyle{bigblock} = [draw, fill=blue!20, rectangle, 
    minimum height=6em, minimum width=8em]
\tikzstyle{medblock} = [draw, fill=blue!20, rectangle, 
    minimum height=4em, minimum width=4em]    
\tikzstyle{mux} = [draw, fill=black!20, rectangle, 
    minimum height=5em, minimum width=0.1em]    
\tikzstyle{smallblock} = [draw, fill=blue!20, rectangle, 
    minimum height=2em, minimum width=3em]
    
\tikzstyle{data_block} = [draw, fill=green!20, rectangle, 
    minimum height=2em, minimum width=3em]
\tikzstyle{ops_block} = [draw, fill=blue!20, rectangle, 
    minimum height=2em, minimum width=3em]    
\tikzstyle{est_block} = [draw, fill=red!20, rectangle, 
    minimum height=2em, minimum width=3em]    
    
\tikzstyle{sum} = [draw, fill=blue!20, circle, node distance=1cm,minimum height=0.5cm]
\tikzstyle{signal} = [coordinate]
\tikzstyle{pinstyle} = [pin edge={to-,thin,black}]
\tikzstyle{block} = [draw, fill=blue!20, rectangle, 
    minimum height=3em, minimum width=9em]
\tikzstyle{blockS} = [draw, fill=blue!20, rectangle, 
    minimum height=3em, minimum width=4em]    
\tikzstyle{input} = [coordinate]
\tikzstyle{output} = [coordinate]
\usetikzlibrary{matrix}

\usetikzlibrary{positioning}
\usetikzlibrary{math} 

\usetikzlibrary{patterns}
\usetikzlibrary{calc,patterns,decorations.pathmorphing,decorations.markings}
\usepackage{xcolor}
\usepackage{hyperref}
\hypersetup{
    colorlinks,
    linkcolor={blue!100!black},
    citecolor={blue!50!black},
    urlcolor={blue!80!black}
}

\newcommand{\bc}{\begin{center}}
\newcommand{\ec}{\end{center}}
\newcommand{\benum}{\begin{enumerate}}
\newcommand{\eenum}{\end{enumerate}}
\newcommand{\nn}{\nonumber}
\newcommand{\matl}{\left[ \begin{array}}
\newcommand{\matr}{\end{array} \right]}

\renewcommand{\matl}{\begin{bmatrix}}
\renewcommand{\matr}{\end{bmatrix}}

\newcommand{\matls}{\left[ \begin{smallmatrix}}
\newcommand{\matrs}{\end{smallmatrix} \right]}
\newcommand{\isdef}{\stackrel{\triangle}{=}}

\newcommand{\inv}{^{-1}}

\newcommand{\dpder}[2]{\displaystyle\frac{\partial {#1}}{\partial {#2}}}

\newcommand{\tr}{{\rm tr}\,}

\newcommand{\rmF}{{\rm F}}

\newcommand{\rmT}{{\rm T}}

\newcommand{\rma}{{\rm a}}

\newcommand{\rmc}{{\rm c}}

\newcommand{\rmr}{{\rm r}}
\newcommand{\rms}{{\rm s}}

\newcommand{\BBE}{{\mathbb E}}

\newcommand{\BBR}{{\mathbb R}}

\newcommand{\SC}{{\mathcal C}}

\newcommand{\SH}{{\mathcal H}}

\newcommand{\SM}{{\mathcal M}}

\newcommand{\SP}{{\mathcal P}}

\newcommand{\SSS}{{\mathcal S}}

\newcommand{\SV}{{\mathcal V}}

\newcommand{\SX}{{\mathcal X}}
\newcommand{\SY}{{\mathcal Y}}

\newcommand{\shiftq}{{\textbf{\textrm{q}}}}

\newcommand{\neweqline}{\ensuremath{\nn \\ &\quad }}

\usepackage{enumitem,amssymb}
\newlist{todolist}{itemize}{2}
\setlist[todolist]{label=$\square$}
\usepackage{pifont}

\usepackage{comment}

\begin{document}


\SetAuthorBlock{Parham Oveissi}{%
    Ph.D. Candidate, \\
    Department of Mechanical Engineering,\\
   University of Maryland, Baltimore County,\\
   1000 Hilltop Circle,
   Baltimore, MD, 21250  \\
   email: parhamo1@umbc.edu} 


\SetAuthorBlock{Ankit Goel\CorrespondingAuthor}{%
Assistant Professor,\\
Department of Mechanical Engineering,\\
   University of Maryland, Baltimore County,\\
   1000 Hilltop Circle,
   Baltimore, MD, 21250  \\
   email: ankgoel@umbc.edu} 

\title{Model-Free Dynamic Mode Adaptive Control for Data-Driven Control Synthesis}

\keywords{
Data-driven control,
adaptive control,
dynamic mode decomposition,
recursive least squares,
model-free control,
online system identification.}
\begin{abstract}
This paper presents a model-free, data-driven control synthesis method called dynamic mode adaptive control (DMAC) for systems whose mathematical models are unavailable or unsuitable for classical control design. 
The proposed approach combines data-driven dynamics approximation with adaptive control synthesis to enable online controller design using measured system data.
DMAC comprises two main components: a dynamics-approximation module and a controller-synthesis module. 
The dynamics approximation module estimates a local linear representation of the system dynamics directly from measurements using a matrix recursive least-squares algorithm with a forgetting factor. 
The estimated dynamics are then used to compute an online stabilizing controller with full-state feedback and integral action. 
Theoretical analysis establishes convergence properties of the recursive dynamics approximation and boundedness of the closed-loop system under the DMAC controller. 
The performance of the proposed method is demonstrated through numerical examples involving representative dynamical systems, including an unstable linear system, the Van der Pol oscillator, and the Burgers' equation. 
Sensitivity studies further demonstrate the robustness of DMAC with respect to both algorithm hyperparameters and variations in system parameters.
\end{abstract}

\date{Version \versionno, \today}

\maketitle 

\section{Introduction}
Many complex engineering systems, such as combustion processes, fluid flows, and high-dimensional structural systems, lack control-oriented mathematical models. 
Although high-fidelity computational models can accurately capture the underlying physics of these systems, they are often too large and computationally intensive for online and real-time control synthesis. 
Consequently, designing controllers that rely on explicit first-principles models can be difficult or impractical in many real-world applications. 
These challenges have motivated the development of data-driven and model-free control methods that synthesize controllers directly from measured data.
Moreover, these challenges are particularly pronounced in mechanical systems with distributed sensing, such as flexible structures and fluid--structure interaction problems, where high-dimensional measurements must be mapped to low-dimensional control objectives.


\paragraph{Literature Review}
Adaptive control methods provide a classical framework for combining online system identification with controller synthesis. 
Approaches such as Model Reference Adaptive Control (MRAC) \cite{astrom1995adaptive, annaswamy1989stable, ioannou2012robust} and self-tuning regulators \cite{aastrom1977theory, aastrom1973self} estimate model parameters online and update control laws accordingly, forming the basis of certainty-equivalent adaptive control \cite{campi1998adaptive}. 

More recently, learning-based adaptive control methods have been developed that integrate system identification with optimal control synthesis.
In particular, adaptive linear-quadratic regulator (LQR) formulations have received significant attention in both learning and control \cite{abbasi2011regret, dean2018regret, mania2019certainty, cohen2019learning}. 
These approaches provide rigorous guarantees, including regret bounds and sample-efficiency results, for linear systems with unknown dynamics.
However, these methods are largely restricted to linear-quadratic settings and emphasize performance metrics such as cumulative cost, rather than stabilization and tracking of nonlinear and high-dimensional systems. 
In addition, many of these methods rely on batch or episodic data to design controllers, resulting in policies that are computed offline rather than continuously adapted in real time.

Optimization-based model-free control methods directly update control parameters using measured data. 
Examples of such methods include Retrospective Cost Adaptive Control (RCAC) \cite{rahman2016tutorial, rahman2017retrospective}, Predictive Cost Adaptive Control (PCAC) \cite{vander2025predictive}, and reinforcement learning (RL) methods \cite{osinenko2022reinforcement, wallace2024continuous,chen2022robust}. 
RCAC admits efficient real-time implementations but can be sensitive to filter design, while PCAC requires solving constrained optimization problems online, which can limit applicability in fast systems. 
Reinforcement learning methods often require large datasets and substantial computational resources, and may lack guarantees on stability and transient performance in continuous-time settings. 

In parallel, data-driven modeling techniques have received considerable attention due to their ability to identify system dynamics directly from data. 
Dynamic Mode Decomposition (DMD) \cite{schmid2010dynamic, tu2013dynamic} provides a framework for extracting low-order linear representations of complex systems from measurement data. 
Extensions such as Dynamic Mode Decomposition with control (DMDc) \cite{proctor2016dynamic, korda2018linear} incorporate actuation into the identified models, while online variants of DMD \cite{zhang2019online} enable adaptation to time-varying dynamics. 
Although these methods are effective for system identification and reduced-order modeling, they are primarily used in an offline setting, and their integration with real-time control synthesis remains limited. 
However, a unified framework that enables online identification of state-space dynamics directly in the measurement space, preserves physical interpretability, and simultaneously performs closed-loop control synthesis remains largely unexplored.

\paragraph{Open Challenges}
Despite significant progress, several challenges remain in the development of data-driven control methods.
First, many existing approaches are limited to linear-quadratic formulations and do not directly address stabilization and tracking objectives for nonlinear, high-dimensional systems. 
Second, several methods rely on batch or episodic data and are not naturally suited for continuous online adaptation. 
Third, optimization-based model-free methods can introduce computational burdens or require careful tuning of design parameters. 
Finally, commonly used input-output identification frameworks may yield internal representations that lack physical interpretability, thereby complicating controller synthesis and limiting insight into the learned dynamics.
These limitations motivate the development of a control framework that directly integrates data-driven state-space identification with real-time control synthesis, while preserving physical interpretability and avoiding reliance on batch data or optimization-based updates.

\paragraph{Contributions}
This paper introduces a data-driven control framework, termed \textit{Dynamic Mode Adaptive Control} (DMAC), that integrates dynamic mode approximation with adaptive control synthesis to enable real-time controller design directly from measured data. 
The proposed approach constructs a local linear approximation of the system dynamics using a matrix recursive least-squares (RLS) algorithm with a forgetting factor, allowing the model to adapt to time-varying or uncertain dynamics. 
The identified dynamics are used directly for controller synthesis, resulting in a computationally efficient architecture suitable for online implementation. 
Unlike classical certainty-equivalent adaptive control, which assumes a known model structure with unknown parameters, the proposed approach constructs the system dynamics directly from data without assuming a parametric model form.


A key feature of the proposed method is its ability to operate in settings where measurements are obtained from a high-dimensional sensing space, while control objectives are defined over a lower-dimensional set of physically meaningful variables. 
Such scenarios arise in applications such as vibration suppression of flexible structures and propulsion systems with distributed sensing. 
The proposed framework preserves a direct relationship between measured variables and the system dynamics, enabling low-order representations suitable for control.
Unlike many data-driven approaches that rely on input-output models, the proposed method identifies a state-space representation directly in terms of measured system variables. 
This structure-preserving formulation retains physical interpretability and enables controller synthesis directly on the identified dynamics.

The main contributions of this paper are summarized as follows.
\begin{enumerate}
\item
A data-driven control architecture (DMAC) that integrates dynamic mode approximation with adaptive control synthesis, enabling real-time controller design directly from streaming measurement data.

\item 
A recursive identification algorithm that estimates a low-order state-space model directly in the measurement space using a matrix-valued RLS formulation, avoiding vectorization, lifting, and input-output model reconstruction.

\item 
A theoretical analysis establishing convergence properties of the identification algorithm and uniform ultimate boundedness of the closed-loop system under the proposed controller.

\item 
Simulation studies on representative dynamical systems, including linear, nonlinear, and distributed-parameter examples, demonstrating stabilization, regulation, and tracking using measurement data alone.
\end{enumerate}

The remainder of the paper is organized as follows. 
Section \ref{sec:DMAC} presents the DMAC algorithm, including the dynamic mode approximation and control synthesis procedure. 
Section \ref{sec:stability} provides the convergence and closed-loop stability analysis. 
Section \ref{sec:exmp} presents numerical examples illustrating the performance of the proposed method. 
Section \ref{sec:conclusions} concludes the paper and outlines directions for future research.
Supporting derivations and technical results are provided in the appendices, including the recursive least-squares formulation underlying the dynamic mode approximation, the integral control design used in the DMAC architecture, and detailed proofs of the theoretical results.

\section{Dynamic Mode Adaptive Control}
\label{sec:DMAC}
This section presents the dynamic mode adaptive control (DMAC) algorithm.

\subsection{Problem Formulation}

Consider a dynamic system $\SM$ whose input is $u(t) \in \BBR^{l_u}$ and the output is $y(t) \in \BBR^{l_y},$ as shown in Figure \ref{fig:DMAC_architecture}.
Let $x(t) \in \BBR^{l_x}$ denote the internal state of $\SM$.
The state $x(t)$ is not assumed to be fully measurable.
Instead, a measured portion of the state is available and denoted by $\xi(t) \in \BBR^{l_\xi}$.

\begin{figure}[h]
    \centering
    \resizebox{1\columnwidth}{!}
    {
    \begin{tikzpicture}[auto, node distance=2cm,>=latex',text centered, line width = 1.5]

        \draw[draw=black, fill=yellow!10] (-2,-2.25)              
             rectangle ++(3.75,3.5) node [xshift=-5.50em, yshift=-1em] {\textbf{DMAC}} ;

        \draw[draw=black, fill=green!10] (2.5,-1.25)              
             rectangle ++(4.75,2.5) node [xshift=-7.em, yshift=-7em] {\textbf{Sampled-data System}} ;
             
        \node [smallblock, blue, fill = blue!20, minimum width=6em, minimum height=3em] (Controller) {Control};

        \node [smallblock, right = 5 em of Controller] (zoh) {ZOH};

        \node [smallblock, red, fill=red!20,right = 2 em of zoh, minimum height=3em] (Plant) {$\SM$};
        
        \node[circle,draw=black, fill=white, inner sep=0pt,minimum size=3pt] (rc11) at ([xshift=5em,yshift=1em]Plant) {};
        \node[circle,draw=black, fill=white, inner sep=0pt,minimum size=3pt] (rc21) at ([xshift=4em,yshift=1em]Plant) {};
        \draw [-] (rc21.north east) --node[below,yshift=.55cm]{$T_\rms$} ([xshift=.3cm,yshift=.15cm]rc21.north east) {};

        \node[circle,draw=black, fill=white, inner sep=0pt,minimum size=3pt] (rc11_xi) at ([xshift=5em,yshift=-1em]Plant) {};
        \node[circle,draw=black, fill=white, inner sep=0pt,minimum size=3pt] (rc21_xi) at ([xshift=4em,yshift=-1em]Plant) {};
        \draw [-] (rc21_xi.north east) --
        ([xshift=.3cm,yshift=.15cm]rc21_xi.north east) {};
        
        \node [smallblock, blue, fill = blue!20, below = 2 em of Controller, minimum width=6em] (DMA) {DMA};

        \draw[<-] (Controller.160) -- +(-2,0) node[xshift = 1em, yshift = 0.75em]{$r_k$};
        \draw[->] (rc11) -- +(2,0) node[xshift = -1em, yshift = 0.75em]{$y_k$};        
        \draw[->] (Controller) node[xshift = 6.5em, yshift = 0.75em]{$u_k$} -- (zoh);
        \draw[->] (zoh) node[xshift = 2.6em, yshift = 0.75em]{$u(t)$} -- (Plant);
        \draw[-] (Plant.32) node[xshift = 1em, yshift = 0.75em]{$y(t)$} -- (rc21) ;
        \draw[-] (Plant.-32) node[xshift = 1em, yshift = 0.75em]{$\xi(t)$} -- (rc21_xi) ;

        \draw[->] (rc11) -| +(1.2,-2.5) |- (-2.5,-2.5) |-(Controller.180) 
        node[xshift = -5.5em, yshift = 0.0em]{$y_k$};
        \draw[->,blue] (rc11_xi) node[xshift = 2em, yshift = 0.75em]{$\xi_k$} -| +(0.75,-1.25) |-(DMA.-10);
        
        \draw[blue,->] (Controller.0) -| +(0.4,-1) node[xshift = -0.75em, yshift = 0.1em]{$u_k$} |- (DMA.10);
        \draw[blue,->] (DMA.180) node[xshift = 0.25em, yshift = 2em]{$A_k, B_k$} -| +(-0.5,1)  |- (Controller.200);

    \end{tikzpicture}
    }
        \caption{Dynamic Mode Adaptive Control (DMAC) architecture for model-free, data-driven, and learning-based control of sampled-data systems.        
        }
        \label{fig:DMAC_architecture}
    \end{figure}
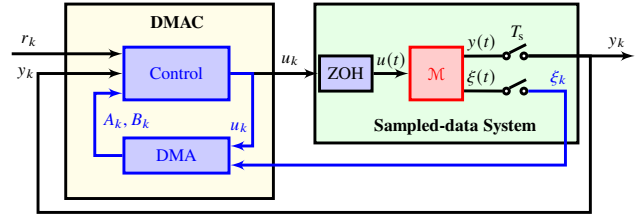
    
Letting $T_\rms>0$ denote the sample time, the system's output is sampled to generate the sampled measurements 
    $y_k \isdef y(k T_\rms). $  
The control input is implemented using a zero-order hold so that, for $t \in [kT_{\rm s},(k+1)T_{\rm s}),$
$u(t) = u_k.$
The measured state signal is similarly sampled to obtain
    $\xi_k \isdef \xi(kT_{\rm s}).$
The objective of the DMAC controller is to generate the discrete-time control input $u_k$ using the measured signals $\xi_k$ and $y_k$ such that the sampled output $y_k$ tracks the reference signal $r_k$.

\subsection{Dynamic Mode Adaptive Control Architecture}

The DMAC architecture is shown in Figure \ref{fig:DMAC_architecture}.
The approach combines online system identification with control synthesis.
At each time step, the DMAC algorithm performs the following two operations.
\begin{enumerate}
\item recursively estimate a local linear model of the system dynamics using the measured signals $\xi_k$ and $u_k$, as described in Section \ref{sec:DynApprox}, and 
\item synthesize a control input that achieves the desired control objective, such as stabilization, reference tracking, or disturbance rejection, as described in Section \ref{sec:controlUpdate}.
\end{enumerate}
This architecture allows the controller to continuously update its internal model of the system dynamics as new data becomes available.
Furthermore, the recursive design enables the controller to adapt to changes in system dynamics without requiring a prior model of the system.

\subsection{Dynamic Mode Approximation}
\label{sec:DynApprox}
%
To synthesize the control input $u_k$, DMAC first estimates a local linear approximation of the system dynamics based on the measured signals $\xi_k$ and $u_k$. 
In particular, DMAC estimates matrices $A \in \BBR^{l_\xi \times l_\xi}$ and $B \in \BBR^{l_\xi \times l_u}$ such that
\begin{align}
    \xi_{k+1} = A \xi_{k} + B u_{k},
    \label{eq:linear_approximation}
\end{align}
which represents a local linear approximation of the dynamics of the measured state.
Note that \eqref{eq:linear_approximation} can be written as
\begin{align}
    \xi_{k+1} = \Theta \phi_k.
    \label{eq:linear_approximation_AxForm_1}
\end{align}
where 
\begin{align}
    \Theta &\isdef \matl A  & B \matr \in \BBR^{l_\xi \times (l_\xi + l_u)}, 
    \\
    \phi_k &\isdef \matl \xi_{k} \\ u_{k} \matr \in \BBR^{l_\xi+l_u}. 
    \label{eq:DMACregressor}
\end{align}
Note that $\Theta$ is the matrix of unknown parameters and $\phi_k$ is the regressor vector. 
The representation \eqref{eq:linear_approximation_AxForm_1} differs from the classical regressor form $y = \phi \theta$, where the unknown parameter $\theta$ is a column vector and the regressor $\phi$ is a matrix. 
In contrast, the unknown parameter in \eqref{eq:linear_approximation_AxForm_1} is the matrix $\Theta$.

Indeed, \eqref{eq:linear_approximation_AxForm_1} can be rewritten in vectorized form as
\begin{align}
    \xi_{k+1} = (\phi_k^\rmT \otimes I)\,\mathrm{vec}(\Theta),
\end{align}
where $\mathrm{vec}(\Theta)$ denotes the stacked parameter vector and $\otimes$ denotes the Kronecker product. 
However, DMAC avoids this reformulation and instead estimates $\Theta$ directly in matrix form. 
This structure-preserving formulation yields the state-space matrices directly, avoids vectorization of the unknown parameters, and leads to a covariance matrix of dimension $(l_\xi+l_u)\times(l_\xi+l_u)$ rather than $l_\xi(l_\xi+l_u)\times l_\xi(l_\xi+l_u)$.

In general, a matrix $\Theta$ that exactly satisfies \eqref{eq:linear_approximation_AxForm_1} may not exist.
Instead, DMAC computes an approximate model by minimizing the exponentially weighted cost function
\begin{align}
    J_k (\Theta)
        \isdef 
            \sum_{i=0}^{k} &\lambda^{k-i} \| \xi_{i} - \Theta \phi_{i-1} \|^2_2 
            +
            \lambda^k \tr (\Theta^\rmT R_\Theta \Theta) ,
    \label{eq:J_k_def}
\end{align}
where
$R_\Theta \in \BBR^{(l_\xi+l_u) \times (l_\xi+l_u)}$ is a positive definite regularization matrix that ensures the existence of the minimizer of \eqref{eq:J_k_def}
and 
$\lambda \in (0,1]$ is a forgetting factor. 
The forgetting factor allows the identification algorithm to emphasize recent data while gradually discounting older measurements.
This property is particularly important for nonlinear or time-varying systems, where a single linear model may only accurately describe the system dynamics locally.
This matrix-valued formulation preserves the state-space structure of the identified model and avoids the dimensional growth associated with vectorized RLS implementations.



\begin{proposition}
    Consider the cost function \eqref{eq:J_k_def}.
    For all $k\geq 0,$ define the minimizer of \eqref{eq:J_k_def} as
    \begin{align}
        \Theta_k 
            \isdef 
                \operatorname*{arg\,min}_{\Theta \in \BBR^{l_\xi \times (l_\xi + l_u)}} J_k(\Theta).
        \label{eq:Theta_k_def}
    \end{align}
    Then, the minimizer $\Theta_k$ satisfies
    \begin{align}
        \Theta_k
            &=
                \Theta_{k-1} 
                +
                \left(
                    \xi_{k} - \Theta_{k-1} \phi_{k-1}
                \right)
                \phi_{k-1}^\rmT \SP_k  
        \label{eq:Theta_k}
            , \\
        \SP_k
            &=
                \lambda \inv \SP_{k-1} 
                -
                \lambda \inv
                \SP_{k-1} \phi_{k-1}
                \Gamma_k \inv 
                \phi_{k-1}^\rmT \SP_{k-1},
        \label{eq:SP_k}
    \end{align}
    where, for all $k \geq 0,$ 
    $\Gamma_k \isdef \lambda  +  \phi_{k-1}^\rmT \SP_{k-1} \phi_{k-1},$ and  
    $\Theta_0 = 0,$
    $\SP_0 \isdef R_\Theta\inv. $
\end{proposition}
\begin{proof}
    See Proposition \ref{prop:theta_k_recursive} in Appendix \ref{appndx:matrix_RLS}.
\end{proof}

The matrices $A_k$ and $B_k$ used in the control design step are obtained from the estimate $\Theta_k$ as
\begin{align}
\Theta_k =
\begin{bmatrix}
A_k & B_k
\end{bmatrix}.
\end{align}

The cost function \eqref{eq:J_k_def} can be viewed as a matrix-valued extension of the recursive least-squares cost commonly used in engineering applications \cite{goel2020recursive}. 
Although \eqref{eq:linear_approximation_AxForm_1} can be rewritten in vectorized form and analyzed using standard RLS arguments, the matrix formulation used here is advantageous because it directly yields the state-space matrices $A_k$ and $B_k$ without parameter stacking or realization reconstruction. 
In addition, the associated covariance matrix $\SP_k$ depends only on the regressor dimension $l_\xi+l_u$, which is typically much smaller than the dimension required by a vectorized parameterization. 
As shown in \cite{Mareels1986,Mareels1988,goel2020recursive}, persistency of excitation is required to ensure that 
(i) the parameter estimate converges and 
(ii) the covariance matrix $\SP_k$ remains bounded.
%
%
To promote persistency of excitation, the control signal includes an \textit{excitation} component, which is designed to ensure persistence of excitation in the regressor $\phi_k$. 
This modification is discussed in Section \ref{sec:controlUpdate}.


\subsection{Connection with Dynamic Mode Decomposition}

The identification component of DMAC is closely related to the dynamic mode decomposition (DMD) method, which is widely used for data-driven modeling and reduced-order analysis of complex dynamical systems. 
This subsection shows the connection between classical DMD and the identification procedure used in DMAC.

Consider a discrete-time system
\begin{align}
    x_{k+1} = f(x_k, u_k), 
    \label{eq:SS_NL}
\end{align}
where $x_k \in \BBR^{l_x}$ is the state vector, $u_k \in \BBR^{l_u}$ is the input vector, and $f : \BBR^{l_x} \times \BBR^{l_u} \rightarrow \BBR^{l_x}$ represents the system dynamics.
For $k \ge 1$, define the state snapshot matrix and the input snapshot matrix
\begin{align}
    X_k
        &\isdef
        \matl 
        x_1 & x_2 & \cdots & x_k
        \matr
        \in \BBR^{l_x \times k},
    \\
    U_k
        &\isdef
        \matl 
        u_1 & u_2 & \cdots & u_k
        \matr
        \in \BBR^{l_u \times k}.
\end{align}
The objective of DMD with inputs is to determine matrices
$A \in \BBR^{l_x \times l_x}$ and
$B \in \BBR^{l_x \times l_u}$ such that
\begin{align}
    X_{k+1} = A X_k + B U_k .
    \label{eq:linear_approximation_DMD}
\end{align}
Equation \eqref{eq:linear_approximation_DMD} can be written in the compact form
\begin{align}
    X_{k+1} = \Theta \SX_k ,
    \label{eq:linear_approximation_AxForm}
\end{align}
where
\begin{align}
    \Theta &\isdef \matl A & B \matr 
        \in \BBR^{l_x \times (l_x + l_u)},
    \\
    \SX_k &\isdef \matl X_k \\ U_k \matr
        \in \BBR^{(l_x + l_u) \times k}.
\end{align}
In classical DMD, the matrix $\Theta$ is then obtained by minimizing the regularized least-squares cost
\begin{align}
    J_{{\rm dmd},k} (\Theta)
        \isdef 
        \| X_{k+1} - \Theta \SX_k \|_\rmF^2
        + \tr (\Theta^\rmT R_\Theta \Theta),
    \label{eq:J_def_DMD}
\end{align}
where $\|M\|_\rmF \isdef \tr(MM^\rmT)$ denotes the Frobenius norm and
$R_\Theta \in \BBR^{(l_x+l_u) \times (l_x+l_u)}$ is a positive definite regularization matrix \cite{strang2022introduction}.

Note that the Frobenius norm term can be expanded as
\begin{align}
\| X_{k+1} - \Theta \SX_k \|_\rmF^2
    =
    \sum_{i=0}^{k}
    \| x_{i+1} - \Theta \chi_i \|_2^2,
\end{align}
where
\begin{align}
    \chi_i \isdef
    \matl
    x_i \\ u_i
    \matr
    \in \BBR^{l_x + l_u}.
\end{align}
Therefore, the cost function minimized in classical DMD is equivalent to the least-squares cost used in the DMAC identification procedure. 
Consequently, the identification step of DMAC can be interpreted as a recursive implementation of dynamic mode decomposition.

\subsection{Control Law Update}
\label{sec:controlUpdate}

This subsection describes the method used to compute the control input $u_k$ using the dynamics approximation obtained in Section \ref{sec:DynApprox}.
To track the reference signal $r_k$, the DMAC algorithm employs a full-state feedback controller with integral action, described in Appendix \ref{sec:FSFi}. 
Note that the term full-state refers to the measured state $\xi_k$ rather than the internal system state $x_k$.

In particular, the control law is given by
\begin{align}
    u_k = K_{\xi,k} \xi_k + K_{q,k} q_k + v_k,
    \label{eq:DMAC_control}
\end{align}
where $K_{\xi,k} \in \BBR^{l_u \times l_\xi}$ and $K_{q,k} \in \BBR^{l_u \times l_y}$ are the time-varying state-feedback and integral gains, respectively. 
These gains are computed using the procedure described in Appendix \ref{sec:FSFi}. 
The control architecture is shown in Figure \ref{fig:FSF_control_loop_int}.
The \textit{excitation signal} $v_k $ is added to the control input to promote persistency of excitation in the regressor $\phi_k$ used in the dynamic mode approximation step.
In particular, Proposition \ref{prop:exciting_input} shows that, in the context of a linear system, an i.i.d., zero-mean input signal ensures persistency of excitation.

\begin{figure}[h]
    \centering
    {
    \begin{tikzpicture}[auto, node distance=2cm,>=latex',text centered]
    
        \node at (-3,0) (reference) {$r$};
        \node[sum, right = 1.5 em of reference] (sum) {};
        \node [smallblock, right = 1.5 em of sum] (integrator) {$\dfrac{\shiftq}{\shiftq-1}$};
        
        \node [smallblock, fill = red!20 ,right = 1.5 em of integrator] (IntGain) {$K_{q,k}$};
        \node[sum, right = 1.5 em of IntGain] (sum1) {};
        
        \node [smallblock, right = 1.5 em of sum1] (Plant) {System};
        \node [smallblock, fill = red!20, below = 1 em of Plant] (FSFG) {$K_{\xi,k}$};
        \node[right = 3.5 em of Plant] (output) {$y$};
        
        \draw[->] (reference) -- (sum) node[xshift = -0.5em, yshift = -1em]{$-$};;
        \draw[->] (sum) -- (integrator) -- (IntGain) 
            node[xshift = -2.2em, yshift = 0.8em]{$q_k$}
            -- (sum1) -- node[xshift = 0em, yshift = 0em]{$u_k$} (Plant)  -- (output);
        \draw[->] (Plant.340) -- +(0.5,0) 
                    |-
                    node[xshift = -0.7em, yshift = 1.4em]{$\xi_k$}
                    (FSFG);

                    
        \draw[->] (FSFG.180) -| (sum1.270) node[xshift = -0.5em, yshift = -0.5em]{$+$};
        
        \draw[->] (Plant.east) -- +(.5,0)  -| +(.75,-1.5)
                    -| (sum.270);
        
    \end{tikzpicture}
    }
        \caption{Time-varying full-state feedback controller with integral action for reference tracking, where the gains $K_{\xi,k}$ and $K_{q,k}$ are updated using the dynamic approximation obtained by DMAC. }
        \label{fig:FSF_control_loop_int}
    \end{figure}
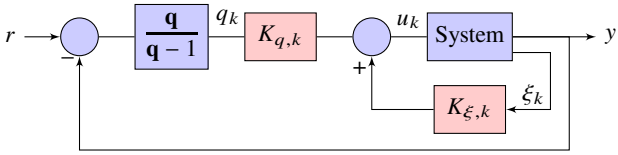

\subsection{DMAC Algorithm}
\label{sec:DMAC_algorithm}

The dynamic mode adaptive control (DMAC) algorithm operates recursively by combining online system identification with control synthesis. 
At each time step $k$, the algorithm updates the local linear approximation of the system dynamics using the measured signals and then computes a control input using the estimated model.

The steps of the DMAC algorithm are summarized as follows.

\begin{enumerate}

\item Measure the signals $\xi_k$ and $y_k$ and compute the reference error $e_k = r_k - y_k$ and the integrator state given by \eqref{eq:integrator_state}.

\item Form the regressor vector \eqref{eq:DMACregressor}.

\item Update the parameter estimate $\Theta_k$ using the recursive least-squares update \eqref{eq:Theta_k}.

\item Extract the matrices
$
\Theta_k =
\begin{bmatrix}
A_k & B_k
\end{bmatrix}.
$

\item Compute the control gains $K_{\xi,k}$ and $K_{q,k}$ using the control design procedure described in Appendix \ref{sec:FSFi}.

\item Compute the control input using \eqref{eq:DMAC_control}.

\end{enumerate}
In summary, the DMAC algorithm recursively estimates a local linear approximation of the system dynamics and uses this approximation to synthesize a stabilizing control input at each time step.

\section{Convergence and Closed-Loop Stability Analysis}
\label{sec:stability}

This section analyzes the theoretical properties of the DMAC algorithm described in Section \ref{sec:DMAC_algorithm}. 
First, we analyze the convergence properties of the recursive least-squares (RLS) algorithm used in the dynamic mode approximation step to estimate the local linear model of the system dynamics. 
Next, we analyze the behavior of the closed-loop system when the estimated dynamics are used for control synthesis. 
In particular, we establish conditions under which the parameter estimates remain well-behaved, and the closed-loop state remains bounded under the DMAC control law.

To facilitate the analysis, we consider the case in which the underlying system dynamics are linear, and the control objective is to stabilize the plant. 
In particular, we consider the discrete-time linear system
\begin{align}
\xi_{k+1} = A^\star \xi_k + B^\star u_k ,
\label{eq:TrueDynamics}
\end{align}
where $A^\star \in \BBR^{l_\xi \times l_\xi}$ and $B^\star \in \BBR^{l_\xi \times l_u}$ are constant matrices such that the pair $(A^\star, B^\star)$ is stabilizable. 
Under this assumption, the dynamic mode approximation described in Section \ref{sec:DynApprox} can recover the true system dynamics under suitable excitation conditions. 
The analysis below establishes convergence properties of the recursive identification step and boundedness of the closed-loop system under the DMAC control law.

\subsection{Convergence of the Recursive Dynamic Mode Approximation}
To analyze the convergence of the dynamic mode approximation, we first express the system dynamics in the regressor form used by the recursive identification algorithm.
In particular, note that \eqref{eq:TrueDynamics} can be written as
\begin{align}
    \xi_{k+1} = \Theta^\star \phi_k.
    \label{eq:dmac_local_model}
\end{align}
where
\begin{align}
    \Theta^\star \isdef [A^\star \  B^\star] \in \BBR^{n \times p}, 
    \quad 
    \phi_k \isdef \matl\xi_k \\ u_k\matr \in \BBR^{p}.
    \label{eq:regressor_def}
\end{align}
The matrix $\Theta$ contains the system matrices $A^\star$ and $B^\star$, and the \textit{regressor} $\phi_k$ contains the measured state and the measured input. 
With the measured regressor, the matrix parameter $\Theta$, at each $k,$ is obtained by minimizing $J_k(\Theta)$ given by \eqref{eq:J_k_def}.
The minimizer of \eqref{eq:J_k_def} is recursively computed by \eqref{eq:Theta_k} and \eqref{eq:SP_k}.

The following definition reviews the concept of persistence of excitation \cite{goel2020recursive}, which is required to establish the stability of the recursive dynamic mode approximation algorithm.

\begin{definition}
[Persistence of excitation]
\label{def:pe}
Let $\{\phi_k\}_{k\ge0}$ be a sequence of regressors with $\phi_k \in \BBR^{p}$.
The sequence $\{\phi_k\}$ is called persistently exciting, if there exist constants $\alpha,\beta>0$ and an integer $N > p$ such that, for all $k\ge 0$
\begin{align}
    \alpha I_{p}  \le  \sum_{i=k}^{k+N} \phi_i \phi_i^\rmT  \le  \beta I_{p}.
    \label{eq:Pk_bounds}
\end{align}

\end{definition}

Under the persistence-of-excitation condition, the covariance matrix generated by the RLS update has the following properties.
\begin{proposition}
\label{prop:P_pos_def}
    Consider the covariance update \eqref{eq:SP_k}.
    Let $\phi_k$ be a persistently exciting regressor. 
    Then, for all $k \ge 0,$ $\SP_k$ satisfies
    \begin{align}
    \label{eq:SPinv_rec}
        \SP_{k+1}\inv = \lambda \SP_{k}\inv + \phi_k \phi_k^\rmT.
    \end{align}
    Furthermore, for all $k\ge0$,  $\SP_k$ is positive definite and bounded.
\end{proposition}
\begin{proof}
    See Section \ref{proof:P_pos_def} in Appendix \ref{appndx:proof}.
\end{proof}

\begin{proposition}
\label{prop:Theta_tilde_update}
Define the matrix parameter error
\begin{align}
    \label{eq:estimation_error}
    \widetilde{\Theta}_k \isdef \Theta_k - \Theta^\star.    
\end{align}
Then, for all $k \geq 0,$  $\widetilde{\Theta}_{k+1}$ satisfies
\begin{align}
    \widetilde{\Theta}_{k+1}
    &=
    \widetilde\Theta_k\big(I - \phi_k\phi_k^\rmT \SP_{k+1}\big)
    \label{eq:est_error_update_v2}
    =
        \lambda \widetilde\Theta_k \SP_{k}\inv \SP_{k+1}.
\end{align}
\end{proposition}
\begin{proof}
    See Section \ref{proof:Theta_tilde_update} in Appendix \ref{appndx:proof}.
\end{proof}

The next result uses Theorems 3, 4, and 5 from \cite{goel2020recursive} to prove stability of the equilibrium $\widetilde \Theta_k=0.$

\begin{theorem}
\label{thm:matrix_rls_stability}
Consider the system
\eqref{eq:est_error_update_v2}-\eqref{eq:SP_k}.
Let $\SP_0$ be positive definite.
If $\lambda =1,$ then the equilibrium of \eqref{eq:est_error_update_v2}, that is, $\widetilde \Theta_k=0$, is Lyapunov stable.
If $\lambda \in (0,1),$ then the equilibrium of \eqref{eq:est_error_update_v2} is globally geometrically stable.
\end{theorem}
\begin{proof}
    See Section \ref{proof:matrix_rls_stability} in Appendix \ref{appndx:proof}.
\end{proof}

Theorem \ref{thm:matrix_rls_stability} shows that the estimation error dynamics associated with the recursive least-squares update are stable.

\subsection{Closed-Loop Boundedness under DMAC}
The convergence properties of the dynamic mode approximation established above are used in the following subsection to analyze the behavior of the closed-loop system under the DMAC controller.
In particular, the next theorem proves the boundedness of the state $\xi_k$ with the proposed DMAC controller in the loop. 

\begin{theorem}
\label{thm:closed_loop_stability_new}

Consider the system \eqref{eq:TrueDynamics},
    where the pair $(A^\star ,B^\star )$ is stabilizable.
    Consider the control law
    \begin{align}
        u_k = K_k \xi_k + v_k,
        \label{eq:dmac_control}
    \end{align}
where
\begin{enumerate}
    \item $(A_k,B_k)$ are estimates of $(A^\star, B^\star)$ computed by the matrix RLS estimator with forgetting factor $\lambda \in (0,1)$, 
    
    \item $K_k$ is the infinite-horizon discrete-time LQR gain associated with 
    the stabilizing solution of the algebraic Riccati equation for the 
    $(A_k,B_k)$ and weighting matrices 
    $Q \ge 0$, $R > 0$,

    \item the excitation signal $v_k$ is i.i.d., zero-mean, with covariance $R_v > 0$, and satisfies the uniform bound
    \begin{align}
        \|v_k\| \le \bar v, 
        \qquad \forall k \ge 0 .
    \end{align} 

\end{enumerate}
Assume that the RLS estimate $(A_k, B_k)$ is stabilizable and $(A_k,Q^{1/2})$ is detectable for all sufficiently large $k,$ so that the corresponding ARE solution exists and $K_k$ is well-defined.
Then there exists a constant $c > 0$, depending on the system and Lyapunov constants, such that the closed-loop state $\xi_k$ is \emph{uniformly ultimately bounded}, and satisfies
\begin{align}
    \limsup_{k \to \infty} \|\xi_k\|
     \le 
    c \, \bar v.
\end{align}

\end{theorem}
\begin{proof}
    See Section \ref{proof:closed_loop_stability_new} in Appendix \ref{appndx:proof}.
\end{proof}

Theorem \ref{thm:closed_loop_stability_new} shows that, under the proposed DMAC architecture, the closed-loop state remains bounded despite the presence of the excitation signal required for system identification.

\section{Numerical Examples}
\label{sec:exmp}
This section illustrates the performance of the proposed DMAC algorithm on three representative dynamical systems of increasing complexity. 
The examples are selected to demonstrate DMAC's ability to stabilize unstable linear systems, regulate nonlinear oscillatory dynamics, and control distributed-parameter systems.
The first example considers an unstable linear system.
This example serves two purposes: it provides numerical verification of the theoretical results developed in Section \ref{sec:DMAC}, and it illustrates the basic operation of the DMAC algorithm.
The second and third examples consider increasingly complex systems, namely the nonlinear Van der Pol oscillator and the Burgers equation.
These examples demonstrate DMAC's ability to regulate nonlinear and distributed-parameter systems via online dynamic mode approximation.


\subsection{Stabilization of an Unstable Linear System}
\label{exmp:linear_unstable}
This example verifies the theoretical properties established in Section \ref{sec:DMAC}.
In particular, the example illustrates the convergence behavior of the recursive dynamic mode approximation and the boundedness of the closed-loop state under the DMAC controller.
Consider the LTI system \eqref{eq:TrueDynamics}, where
\begin{align}
    A^\star &=
    \begin{bmatrix}
        1.05 & 0.25\\
        -0.1 & 0.98
    \end{bmatrix},
    \quad 
    B^\star =
    \begin{bmatrix}
        0.12\\
        0.25
    \end{bmatrix}.
    \label{eq:unstableLinear}
\end{align}
Note that the eigenvalues of $A^\star $ are $ \{1.015 \pm 0.1542 \jmath \},$ and thus the open-loop system is unstable. 
The objective is to stabilize the unstable system.

\textbf{DMAC Setup.}
In this example, we assume that the full state is available for controller synthesis. 
Since $\xi_k \in \BBR^2$ and $u_k \in \BBR,$ it follows that $\Theta_k$ is a $2 \times 3$ matrix. 
In DMAC, we set 
$\Theta_0 = 0_{2\times 3}, 
    \SP_0 = 10^3 I_3, $ and $
    \lambda = 0.995,$
to estimate $\Theta_k=[A_k\ \ B_k] $ using the recursive DMA.
The control gain $K_k$ is computed using LQR with weights
$Q_{\mathrm{lqr}} = I_2, 
R_{\mathrm{lqr}} = 0.2.$
The excitation signal is generated from a uniform distribution. 
In particular, $v_k \sim \mathrm{Unif}([-\,\bar v,\bar v]).$
Finally, the control signal is given by \eqref{eq:dmac_control}.

\textbf{Stabilization.}
Figure \ref{fig:DMAC_closedloop} shows the closed-loop response of the system with the control signal generated by DMAC with the initial condition $\xi_0 = [1 \ -0.5]^\rmT$ and excitation signal $\bar v = 0.01$.
In particular, 
\ref{fig:DMAC_closedloop}a) shows the state $\xi_k$, 
\ref{fig:DMAC_closedloop}b) shows the input $u_k$ generated by DMAC, 
\ref{fig:DMAC_closedloop}c) shows the state norm of the state $\xi_k$, and 
\ref{fig:DMAC_closedloop}d) shows the parameter estimate $\Theta_k$ computed by DMAC.

Note that the states do not converge to zero due to the presence of the exogenous excitation signal $v_k.$
Specifically, once the system is stabilized, the excitation signal acts as an external input to an asymptotically stable system. 
Since the closed-loop system is asymptotically stable, the mapping from $v_k$ to $\xi_k$ is mean-square stable \cite{ZhouDoyleGlover1996,khalil2002nonlinear}. 
Consequently, if $v_k$ is a zero-mean stochastic process with finite variance, then $\xi_k$ is also a zero-mean stochastic process with bounded variance satisfying
\begin{align}
    \sup_k \BBE \left[ \|\xi_k\|^2 \right]
    \le
    \| G_{\xi v} \|_{\SH_\infty}^2 \,
    \sup_k \BBE\left[ \|v_k\|^2 \right].
    \label{eq:SS_inequality}
\end{align}

\begin{figure}[H]
    \centering
    \includegraphics[width=\columnwidth]{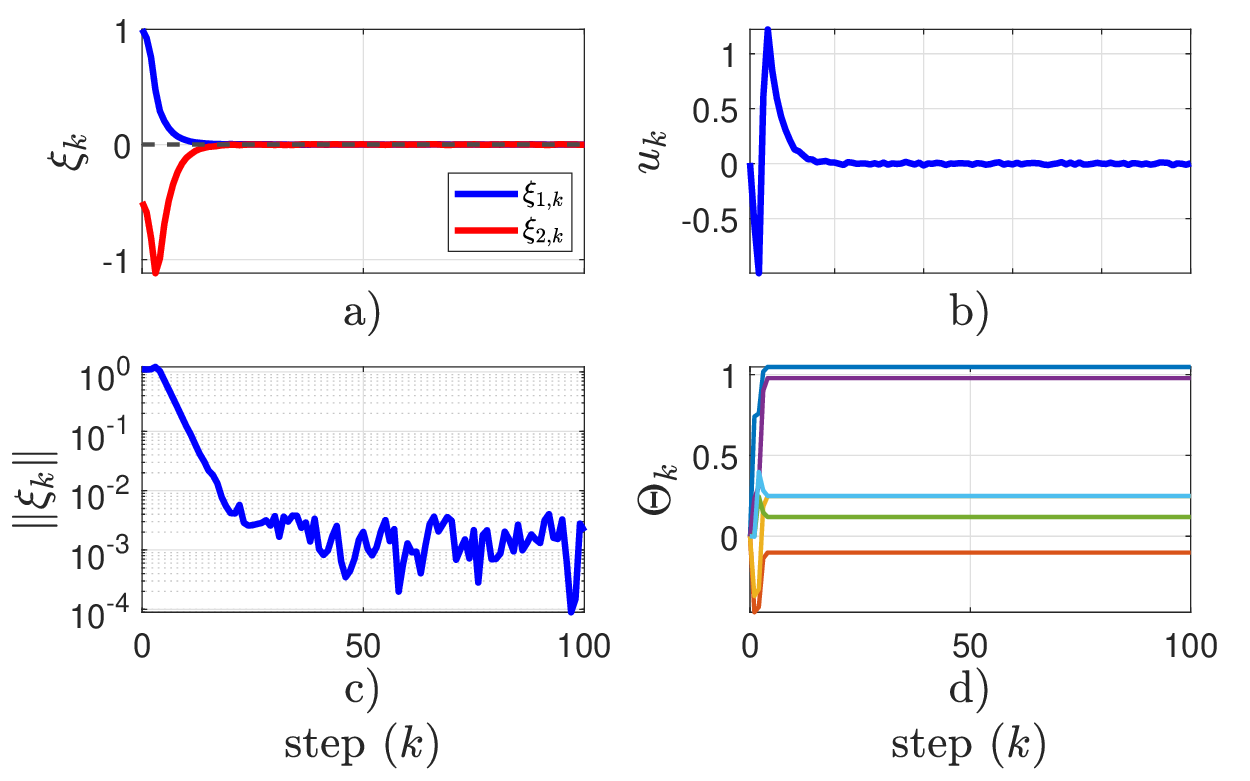}
    \caption{Closed-loop response of the unstable linear system \eqref{eq:TrueDynamics} with the dynamics and input matrices given by \eqref{eq:unstableLinear}.  } 
    \label{fig:DMAC_closedloop}
\end{figure}

Next, Figure \ref{fig:Phase_Portrait_Loop} shows closed-loop trajectories from multiple initial conditions.
In all cases, the DMAC controller stabilizes the open-loop unstable system, and the state converges to a bounded neighborhood of the origin.
\begin{figure}[h]
    \centering
    \includegraphics[width=\columnwidth]{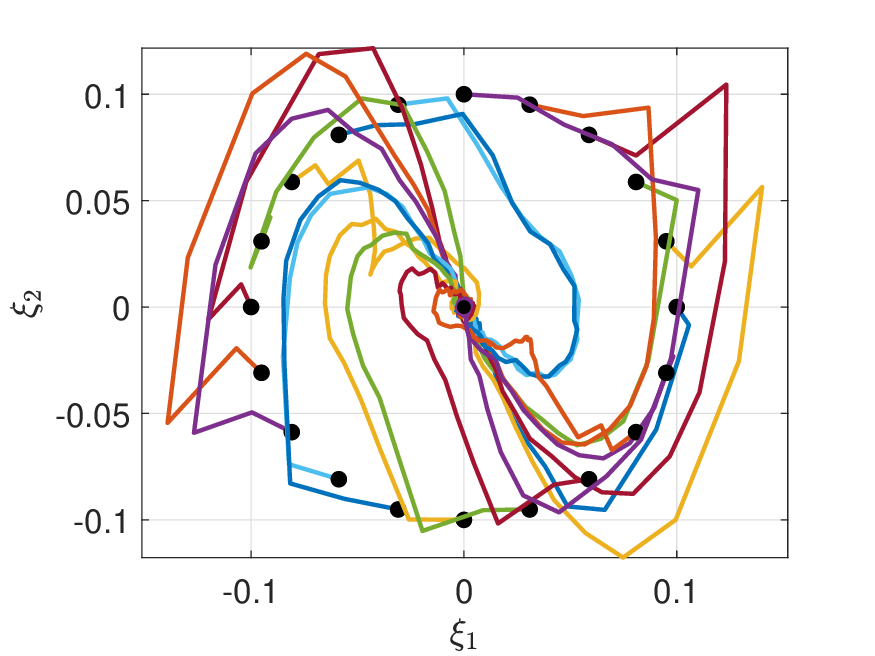}
    \caption{Closed-loop phase portrait with DMAC for various initial conditions. 
    The DMAC controller stabilizes the unstable system, and the state converges to a neighborhood of the origin.} 
    \label{fig:Phase_Portrait_Loop}
\end{figure}

Figure \ref{fig:Terms_Bound} compares the terms $T_{1,k}$-$T_{4,k}$ in the Lyapunov difference expansion \eqref{eq:state_diff_full} with their corresponding analytical upper bounds derived in the proof of Theorem \ref{thm:closed_loop_stability_new}. 
The results verify the validity of the derived inequalities and the collected bound in \eqref{eq:state_bound_collected_raw}.
\begin{figure}[h]
    \centering
    \includegraphics[width=\columnwidth]{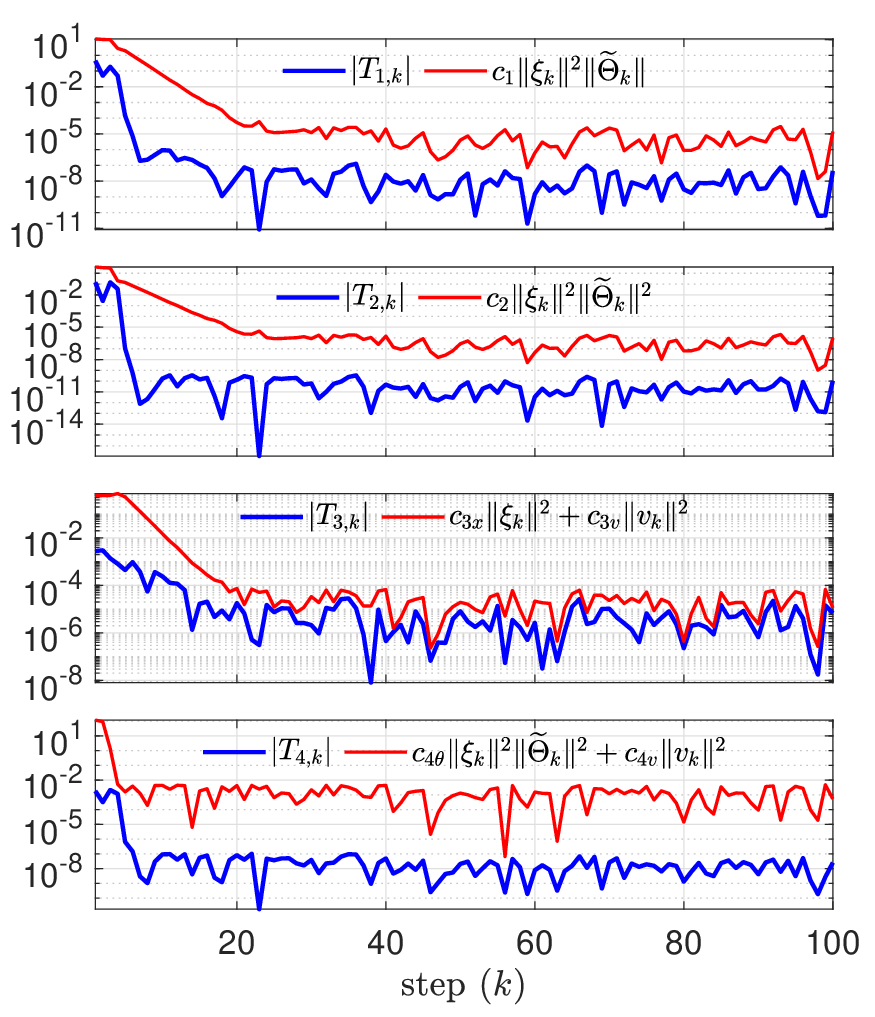}
    \caption{Comparison of the terms $T_{1,k}$--$T_{4,k}$ in the Lyapunov difference expansion \eqref{eq:state_diff_full} with their corresponding analytical upper bounds derived in the proof of Theorem \ref{thm:closed_loop_stability_new}. The results confirm the validity of the derived inequalities and the collected bound in \eqref{eq:state_bound_collected_raw}.}
    \label{fig:Terms_Bound}
\end{figure}

\textbf{State, Estimation Error, and Gain Error Convergence.}
Figure \ref{fig:State_norm_Theta_error_K_error_nomrs} shows the state norm $\|\xi_k\|$, the parameter estimation error $\|\widetilde{\Theta}_k\|_F$, and the gain error $\|\widetilde{K}_k\|_2$ on a logarithmic scale for $\bar v = 0.01$ and $\bar v = 0.05$. 
A smaller excitation magnitude leads to slower convergence of the parameter and gain estimates, which is consistent with the well-known property that the convergence rate of RLS depends on the magnitude of the regressor. 
At the same time, a smaller excitation signal results in convergence of the state to a smaller neighborhood of the origin.
This behavior follows from the asymptotic stability of the closed-loop system, that is, in steady state, the state magnitude scales with the input magnitude according to \eqref{eq:SS_inequality}. 
Consequently, a larger excitation signal induces a larger steady-state response.

\begin{figure}[h] 
    \centering \includegraphics[width=\columnwidth]{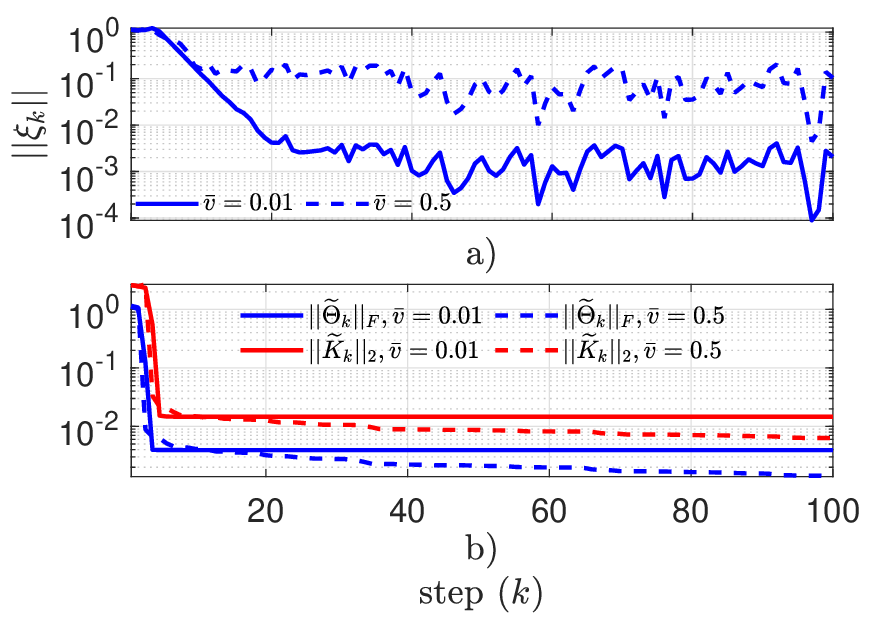} 
    \caption{State norm $\|\xi_k\|$ and parameter estimation and gain errors  $\|\widetilde{\Theta}_k\|_F$ and $\|\widetilde{K}_k\|_2$ on a logarithmic scale for $\bar v = 0.01$ and $\bar v = 0.5$.}
    \label{fig:State_norm_Theta_error_K_error_nomrs} 
    \end{figure}

These results verify the theoretical properties established in Section \ref{sec:DMAC} and demonstrate DMAC's ability to stabilize unstable systems while simultaneously identifying the system dynamics online.

\subsection{Output Tracking in Van Der Pol Oscillator}
This example illustrates the performance of the DMAC controller on a nonlinear system. 
In particular, the example demonstrates DMAC's ability to track reference commands for nonlinear dynamics without requiring a system model.

Consider the Van Der Pol oscillator
\begin{align}
    \label{eq:VDP_ODE} 
     \ddot q -\mu (1 - q^2) \dot q + q = u.
\end{align}
MATLAB's \href{https://www.mathworks.com/help/matlab/ref/ode45.html}{ode45} routine is used to simulate \eqref{eq:VDP_ODE}.
In this example, we set $\mu = 1$ and the initial condition $x(0)$ is randomly generated using MATLAB's \href{https://www.mathworks.com/help/matlab/ref/randn.html}{randn} routine.

The control objective is to ensure that the output
\begin{align}
y_k = q(kT_{\rm s})
\end{align}
tracks a reference signal $r$.
In this example, the DMAC algorithm updates the control signal $u_k$ every $T_\rms = 0.1 $ $\rms.$

\textbf{DMAC Setup.}
To apply the DMAC algorithm, we assume that the measured state is
\begin{align}
\xi_k =
\begin{bmatrix}
q(kT_{\rm s}) \\
\dot q(kT_{\rm s})
\end{bmatrix}.
\end{align}

Since $\xi_k \in \BBR^2$ and $u_k \in \BBR$, the parameter matrix $\Theta_k$ has dimension $2\times3$.
In DMAC, we set 
$\Theta_0 = 0_{2\times 3}, 
    \SP_0 = 10^{-2} I_3, $ and $
    \lambda = 0.995,$
to estimate $\Theta_k=[A_k\ \ B_k] $ using the recursive DMA.
The control gain $K_k$ is computed using the full-state feedback with integral action described in Section \ref{sec:FSFi}, with LQR weights
$Q_{\mathrm{lqr}} = I_3,$ and $ 
R_{\mathrm{lqr}} = 1.$

\textbf{Closed-Loop Response.}
Figure \ref{fig:VDP_DMAC_FSFI} shows the closed-loop response of the Van der Pol system with the DMAC controller. 
In particular, (a) shows the output $y_k$ and the reference signal $r$, 
(b) shows the control signal $u_k$ computed by DMAC,
(c) shows the absolute value of the tracking error $z_k \isdef y_k - r$ on a logarithmic scale, and 
(d) shows the parameter matrix $\Theta_k$ estimated by DMAC.
Note that the output error converges to a neighborhood of zero whose size is proportional to the magnitude of the excitation signal $\bar v$ and remains bounded.
\begin{figure}[h]
    \centering
    \includegraphics[width=\columnwidth]{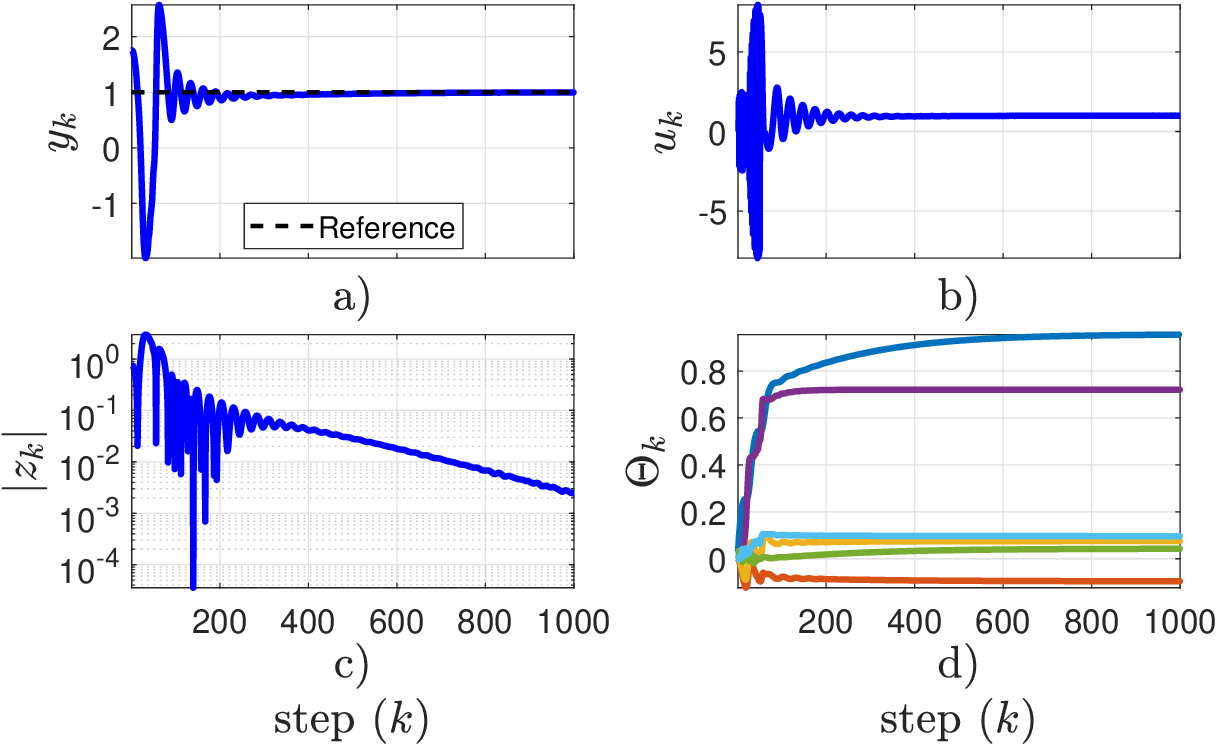}
    \caption{Closed-loop response of \eqref{eq:VDP_ODE} with DMAC. a) shows the output $y_k$ and the reference signal $r,$ b) shows the control signal $u_k,$ c) shows the absolute value of the tracking error $z_k$ on a logarithmic scale, and d) shows the estimate matrix $\Theta_k$ computed by DMAC.} 
    \label{fig:VDP_DMAC_FSFI}
\end{figure}

\textbf{Sensitivity to Algorithm Hyperparameters.}
To investigate the robustness of the DMAC algorithm with respect to its tuning parameters, we vary the hyperparameters $\lambda$, $R_\Theta$, $Q_{\mathrm{lqr}}$, and $R_{\mathrm{lqr}}$, while keeping the remaining parameters at their nominal values.

Figure \ref{fig:VDP_DMAC_Sensitivity_Hyperparameters} shows the effect of these hyperparameters on the closed-loop response $y_k$. 
Specifically, (a), (b), (c), and (d) show the effect of varying $\lambda$, $R_\Theta$, $Q_{\mathrm{lqr}}$, and $R_{\mathrm{lqr}}$, respectively.
In each case, the hyperparameters are varied over several orders of magnitude. 
The results indicate that the DMAC algorithm maintains stable tracking performance across a wide range of parameter values, suggesting robustness to tuning.

\begin{figure}[h]
    \centering
    \includegraphics[width=\columnwidth]{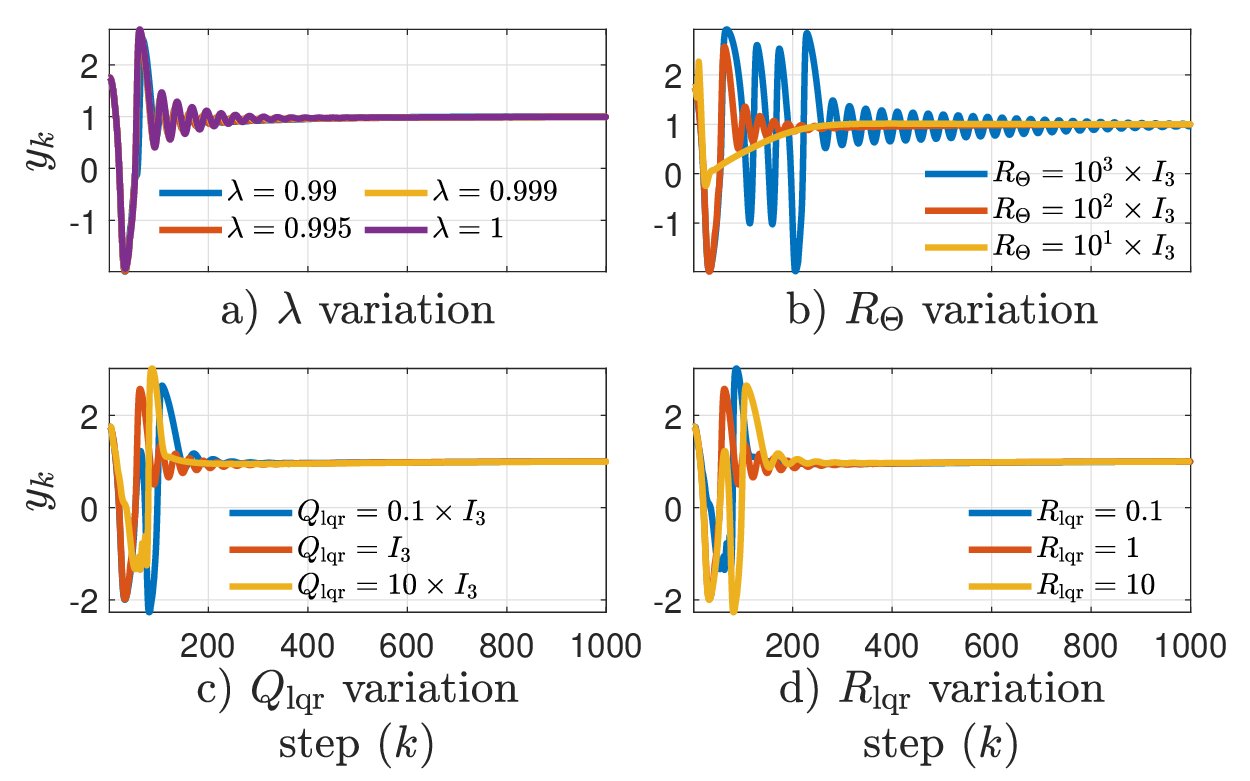}
    \caption{Effect of DMAC hyperparameters on the closed-loop response of the Van der Pol oscillator. 
    a) Effect of the forgetting factor $\lambda$, 
    b) effect of the regularization matrix $R_\Theta$, 
    c) effect of the control weighting matrix $Q_{\mathrm{lqr}}$, and 
    d) effect of the control weighting scalar $R_{\mathrm{lqr}}$. 
    The results indicate that the DMAC controller maintains stable tracking performance over a wide range of hyperparameter values.}
    \label{fig:VDP_DMAC_Sensitivity_Hyperparameters}
\end{figure}

\textbf{Robustness to System Parameters.}
Finally, we assess the robustness of the DMAC controller to variations in the physical system parameters. 
In particular, we vary the nonlinear damping parameter $\mu$ while keeping all other parameters fixed.

Figure \ref{fig:VDP_DMAC_System_Parameters} shows the resulting closed-loop response $y_k$ for different values of $\mu$. 
The result indicates that the DMAC controller maintains stable tracking performance, demonstrating robustness with respect to variations in the system dynamics.

\begin{figure}[h]
    \centering
    \includegraphics[width=\columnwidth]{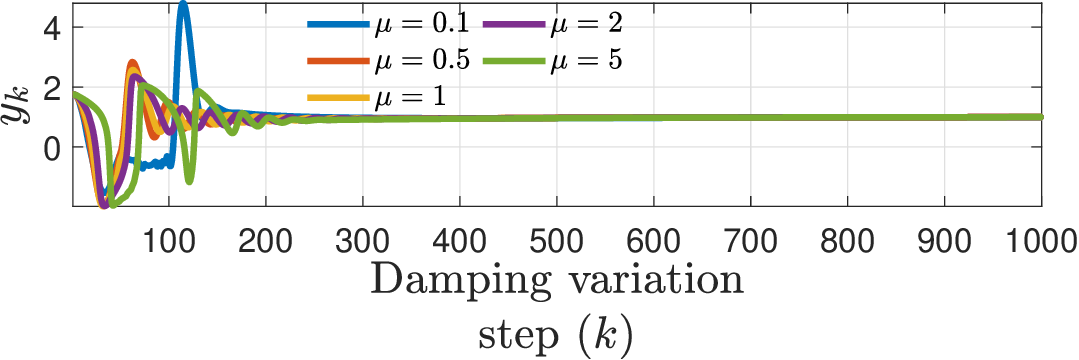}
    \caption{Effect of varying the Van der Pol parameter $\mu$ on the closed-loop response under DMAC control. 
    Despite variations in the system dynamics, the DMAC controller maintains stable tracking performance, indicating robustness to plant parameter uncertainty.}
    \label{fig:VDP_DMAC_System_Parameters}
\end{figure}

\subsection{Output Tracking in Burgers' Equation}
This example demonstrates DMAC's ability to track reference commands for systems governed by nonlinear partial differential equations without requiring a system model.
Consider the one-dimensional Burgers' equation given by
\begin{align}
\label{eq:Burgers_PDE}
\frac{\partial w}{\partial t} + w \frac{\partial w}{\partial x} = \nu \frac{\partial^2 w}{\partial x^2} + u_\rmc(x,t),
\end{align}
where $w(x,t)$ may represent a pressure or a velocity field, $\nu$ is the viscosity coefficient, and $u_\rmc(x,t)$ is the control input applied at the location $x$ at time $t$.

To simulate the Burgers equation \eqref{eq:Burgers_PDE},
we discretize \eqref{eq:Burgers_PDE} in space using a uniform grid with $N$ nodes over the domain $x \in [0, 2\pi].$
Thus, the spatial step size $\Delta x = 2\pi / (N-1)$ and 
the $i$th spatial node is $x_i = (i-1) \Delta x.$
The approximation of the field variable $w(x,t)$ at the spatial node $x_i$ is denoted by $w_i(t).$
At each node, the convective and diffusive terms are then approximated using central differences, that is, 
\begin{align}
    \dpder{w}{x} \Bigg |_{x_i}
        &\approx
            \frac{w_{i+1} -w_{i-1}}{2 \Delta x} 
    \\
    \frac{\partial^2 w}{\partial x^2} \Bigg|_{x_i}
        &\approx
            \frac{w_{i+1} - 2 w_i + w_{i-1}}{\Delta x^2}.
\end{align}
%
Thus, for $i = 1,2,3, \ldots, N,$ 
\begin{align}
\dot{w}_i 
    &=
        - w_i \frac{w_{i+1} - w_{i-1}}{2 \Delta x} 
        \neweqline
        + \nu \frac{w_{i+1} - 2 w_i + w_{i-1}}{\Delta x^2} + u_{\rmc,i}(t). 
    \label{eq:discretizedODE}
\end{align}
To ensure a continuous flow, we impose periodic boundary conditions, where the function values wrap around the domain.
In particular, the periodic boundary conditions are imposed by setting  $w_{0} = w_N$ and $w_{N+1} = w_1.$

In this work, we set $N=100,$ $\nu = 0.1,$ and initialize $w(x,0)$ randomly using MATLAB's \href{https://www.mathworks.com/help/matlab/ref/randn.html}{randn} routine. 
MATLAB's \href{https://www.mathworks.com/help/matlab/ref/ode45.html}{ode45} routine is used to simulate \eqref{eq:discretizedODE}.
The output of the system is assumed to be the field variable at the $61$st node, that is, $y = w_{61}.$
The objective of the DMAC controller is to ensure that the output tracks a unit step reference signal $r.$
In this example, the DMAC algorithm updates the control signal $u_k$ every $T_\rms = 0.01 $ $\rms.$

To reflect the physical scenario with limited sensors, we assume that the field variable $w_i$ is measured at only a few sparse locations. 
In particular, we assume that $w_i$ is measured at the nodes $i \in \{1, 16, 31, 46, 61, 76, 91 \}.$

\textbf{DMAC Setup.} 
To apply the DMAC algorithm, we assume that the measured state is $\xi_k = \matl w_1 & w_{16} & w_{31} & w_{46} & w_{61} & w_{76} & w_{91} \matr^\rmT \Big |_{t=k T_\rms}.$
Finally, in this example, the control is applied at $55$th node, therefore, $u = u_{\rmc,55}.$



Since $\xi_k \in \BBR^7$ and $u_k \in \BBR$, the parameter matrix $\Theta_k$ has dimension $7\times8$.
In DMAC, we set 
$\Theta_0 = 0_{7\times 8}, 
    \SP_0 = 10^{-2} I_8, $ and $
    \lambda = 0.9995,$
to estimate $\Theta_k=[A_k\ \ B_k] $ using the recursive DMA.
The control gain $K_k$ is computed using the full-state feedback with integral action described in Section \ref{sec:FSFi}, with LQR weights
$Q_{\mathrm{lqr}} = 10 \times I_8,$ and $ 
R_{\mathrm{lqr}} = 0.1.$


\textbf{Closed-Loop Response.} Figure \ref{fig:Burgers_DMAC_FSFI} shows the closed-loop response of the Burgers' system \eqref{eq:Burgers_PDE} with the DMAC controller. In particular, (a) shows the output $y_k$ and the reference signal $r$, 
(b) shows the control signal $u_k$ computed by DMAC,
(c) shows the absolute value of the tracking error $z_k \isdef y_k - r$ on a logarithmic scale, and 
(d) shows the parameter matrix $\Theta_k$ estimated by DMAC. Note that the output error converges to a neighborhood of zero whose size is proportional to the magnitude of the excitation signal $\bar v$ and remains bounded. Figure \ref{fig:Burgers_DMAC_states_contour} shows the closed-loop response of the discretized state of the Burgers equation.  
 
\begin{figure}[h]
\centering
\includegraphics[width=\columnwidth]{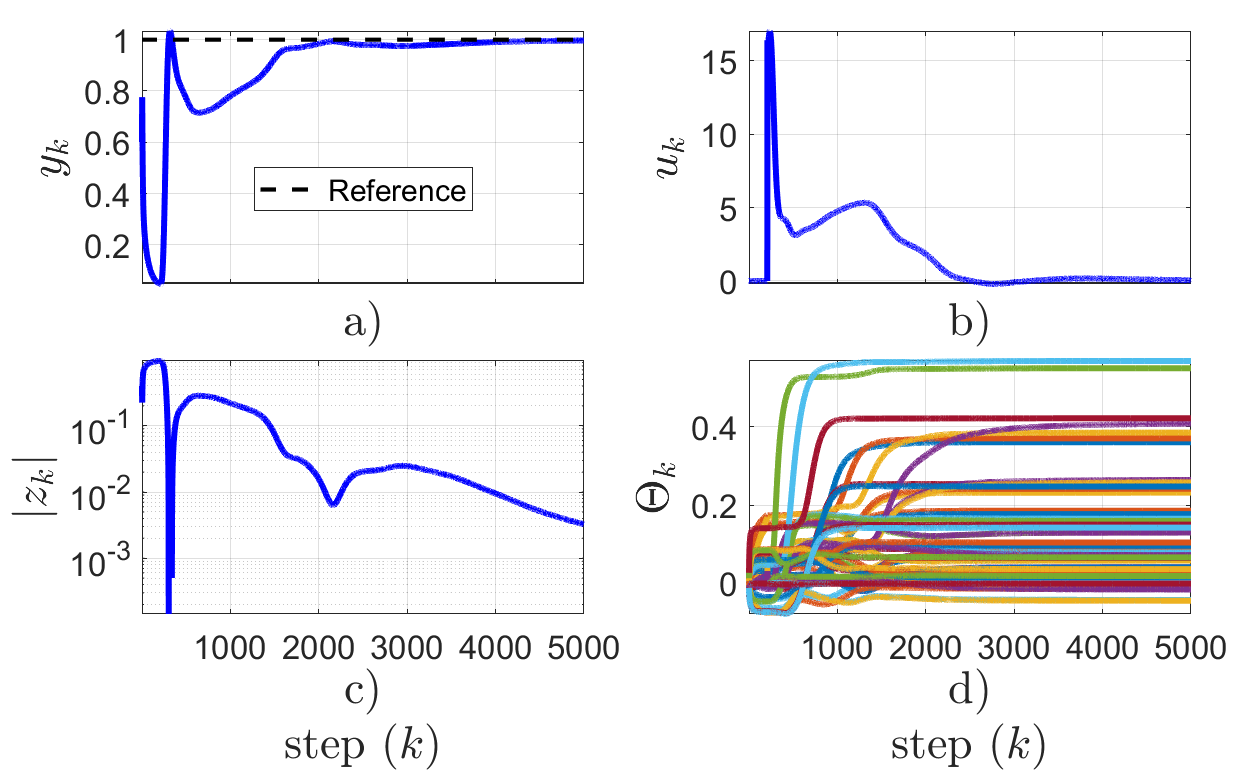}
\caption{Closed-loop response of \eqref{eq:Burgers_PDE} with DMAC. a) shows the output $y_k$ and the reference signal $r,$ b) shows the control signal $u_k,$ c) shows the absolute value of the tracking error $z_k$ on a logarithmic scale, and d) shows the estimate matrix $\Theta_k$ computed by DMAC.}
\label{fig:Burgers_DMAC_FSFI}
\end{figure}

\begin{figure}[H]
\centering
\includegraphics[width=\columnwidth]{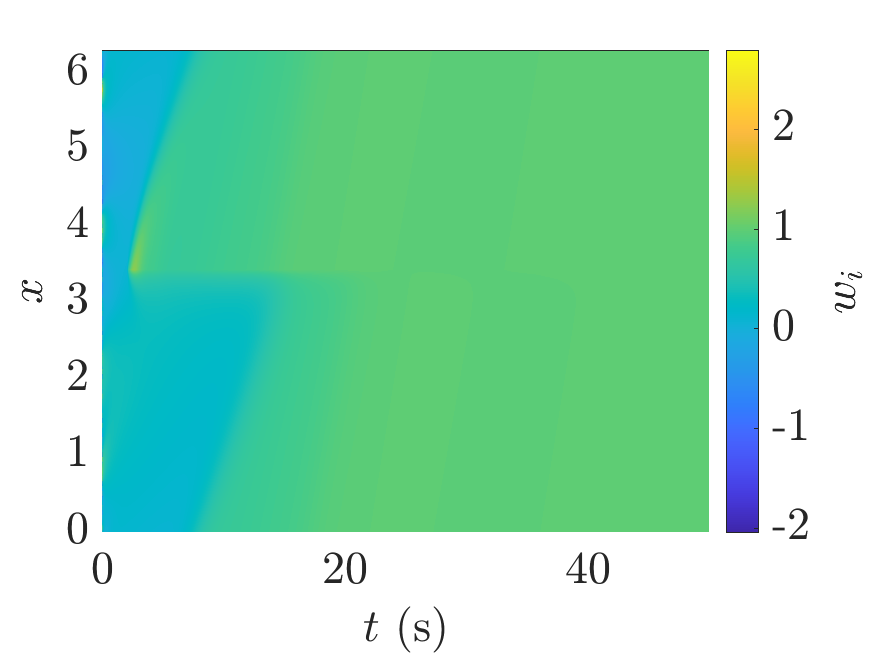}
\caption{Closed-loop response of the discretized Burgers equation.}
\label{fig:Burgers_DMAC_states_contour}
\end{figure}

\textbf{Sensitivity to Algorithm Hyperparameters.}
To investigate the robustness of the DMAC algorithm with respect to its tuning parameters, we vary the hyperparameters $\lambda$, $R_\Theta$, $Q_{\mathrm{lqr}}$, and $R_{\mathrm{lqr}}$, while keeping the remaining parameters at their nominal values.

Figure \ref{fig:Burgers_DMAC_Sensitivity_Hyperparameters} shows the effect of these hyperparameters on the closed-loop response $y_k$. 
Specifically, (a), (b), (c), and (d) show the effect of varying $\lambda$, $R_\Theta$, $Q_{\mathrm{lqr}}$, and $R_{\mathrm{lqr}}$, respectively.
In each case, the hyperparameters are varied over several orders of magnitude. 
The results indicate that the DMAC algorithm maintains stable tracking performance across a wide range of parameter values, suggesting robustness to tuning.
\begin{figure}[h]
\centering
\includegraphics[width=\columnwidth]{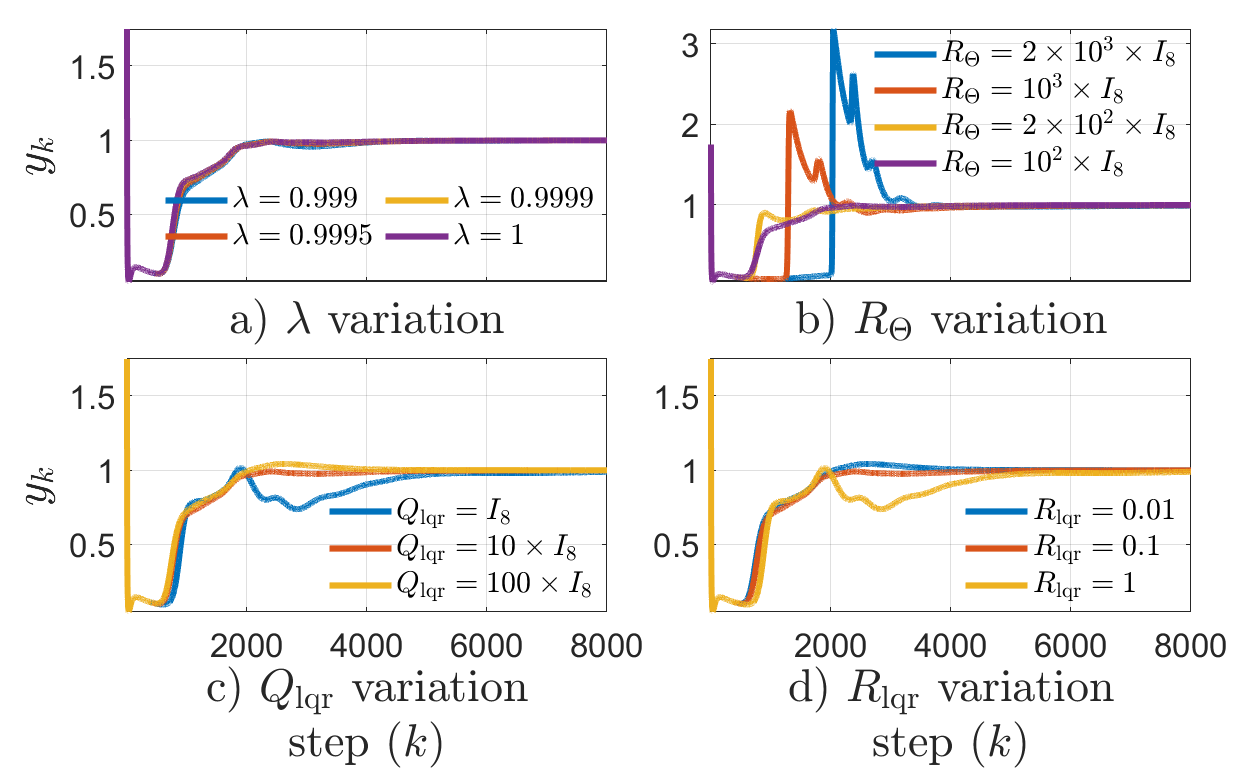}
\caption{Effect of DMAC hyperparameters on the closed-loop performance.}
\label{fig:Burgers_DMAC_Sensitivity_Hyperparameters}
\end{figure}

\textbf{Robustness to System Parameters.}
Finally, we assess the robustness of the DMAC controller to variations in the physical system parameters. 
In particular, we vary the viscosity coefficient $\nu$ while keeping all other parameters fixed.

Figure \ref{fig:Burgers_DMAC_System_Parameters} shows the resulting closed-loop response $y_k$ for different values of $\nu$. 
The result indicates that the DMAC controller maintains stable tracking performance, demonstrating robustness with respect to variations in the system dynamics.
\begin{figure}[h]
\centering
\includegraphics[width=\columnwidth]{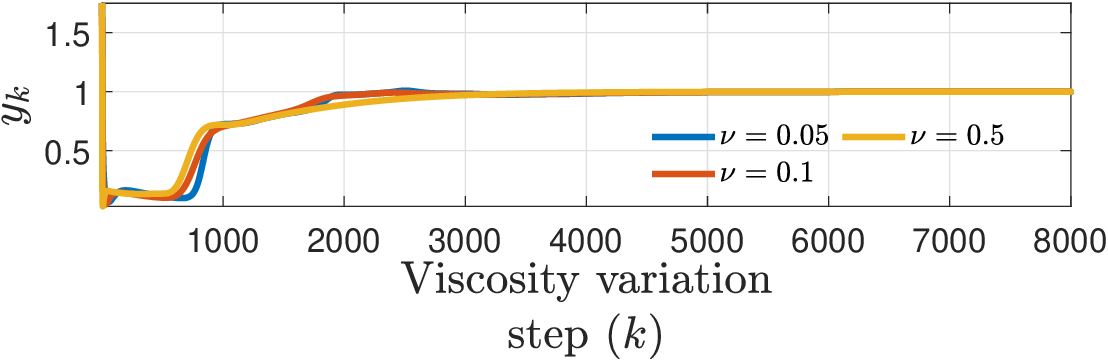}
\caption{Effect of varying the system's physical parameters on the closed-loop performance.}
\label{fig:Burgers_DMAC_System_Parameters}
\end{figure}

\section{Conclusions}
\label{sec:conclusions}

This paper presented a data-driven adaptive control framework called dynamic mode adaptive control (DMAC) for synthesizing controllers for dynamical systems whose mathematical models are unavailable or unsuitable for classical control design. 
The proposed approach integrates dynamic mode approximation with adaptive control synthesis within a unified recursive architecture, enabling controllers to be constructed directly from streaming measurement data.

The identification component of DMAC uses a matrix recursive least-squares algorithm with a forgetting factor to estimate a local linear state-space representation of the system dynamics online. 
Because the identified model evolves with respect to measured system variables, the resulting state representation retains a clear physical interpretation and can be used directly for controller synthesis. 
Theoretical analysis established convergence properties of the recursive identification algorithm and boundedness of the closed-loop system under the DMAC controller.

The effectiveness of the proposed framework was demonstrated through numerical studies on dynamical systems of increasing complexity. 
An unstable linear system was used to verify the theoretical properties of the algorithm and illustrate the convergence behavior of the recursive identification step. 
Additional examples involving the Van der Pol oscillator and the Burgers equation demonstrated DMAC's ability to regulate nonlinear and distributed-parameter systems using limited measurement data. 
Sensitivity studies further indicated that the proposed approach is robust to variations in algorithm hyperparameters and system parameters.

Future work will focus on several directions. 
First, methods will be investigated to reduce or eliminate the need for persistent excitation while maintaining reliable parameter convergence in the recursive dynamics approximation. 
In particular, adaptive excitation mechanisms and identification strategies will be explored to ensure sufficient information for identification while minimizing the impact of excitation signals on closed-loop performance. 
Second, constrained identification techniques will be developed to enforce structural properties in the learned dynamics, such as stabilizability or controllability, thereby improving the robustness and reliability of the resulting controller. 
Finally, future work will investigate real-time implementation of the DMAC framework and evaluate its computational performance and control effectiveness on physical dynamical systems.





\appendix

\section*{Appendix}

This appendix presents the technical derivations and supporting results omitted from the main text for clarity. 
A summary of all appendix results is provided in Table \ref{tab:appendix_summary} to facilitate readability and highlight their role in the analysis.

\begin{table*}[t]
\centering
\caption{Summary of formal results in the appendices.}
\label{tab:appendix_summary}
\begin{tabular}{p{0.2\textwidth} p{0.05\textwidth} p{0.08\textwidth} p{0.55\textwidth}}
\toprule
\textbf{Appendix Section} & \textbf{Label} & \textbf{Type} & \textbf{Description} \\
\midrule

\textit{Matrix RLS} 
& \ref{appndx:matrix_RLS}
& Algorithm 
& Recursive exponentially weighted least-squares update used to estimate the local linear model $\Theta_k = \matl A_k & B_k \matr$, including covariance update $\SP_k$. \\

\midrule

\textit{Full-State Feedback Control with Integral Action} & \ref{sec:FSFi} & Section &
Augmented-state formulation and state-feedback with integral action used in DMAC. \\

\midrule

\textit{Proof of Proposition \ref{prop:P_pos_def}} & \ref{proof:P_pos_def} & Proof &
Proof that the covariance matrix update admits the inverse recursion and that $\SP_k$ remains positive definite. \\


\textit{Proof of Proposition \ref{prop:Theta_tilde_update}} & \ref{proof:Theta_tilde_update} & Proof &
Proof of the matrix estimation error recursion for $\widetilde{\Theta}_k$. \\


\textit{Proof of Theorem \ref{thm:matrix_rls_stability}} & \ref{proof:matrix_rls_stability} & Proof &
Lyapunov proof for stability of the matrix RLS estimation error dynamics. \\


\textit{Proof of Theorem \ref{thm:closed_loop_stability_new}} & \ref{proof:closed_loop_stability_new} & Proof &
Composite Lyapunov proof establishing boundedness of the closed-loop state under DMAC. \\


\textit{Young's Inequality with Parameter} & \ref{lem:young} & Lemma &
Auxiliary inequality used to bound cross terms in the closed-loop Lyapunov analysis. \\

\midrule



\textit{Persistence of excitation under random inputs} & \ref{prop:exciting_input} & Proposition &
Shows that, for a Schur system with controllable $(A,B)$ and i.i.d.\ zero-mean inputs, the regressor is persistently exciting in expectation. \\

\textit{Continuity and local Lipschitz property of the LQR gain} & \ref{prop:lqr_gain_convergence_bound} & Proposition &
Establishes continuity of the infinite-horizon LQR gain and a local Lipschitz bound with respect to perturbations in the pair $(A,B)$. \\

\bottomrule
\end{tabular}
\end{table*}

\section{Recursive Solution of the Exponentially Weighted Regularized Least-Squares Problem}

\label{appndx:matrix_RLS}

This appendix derives the recursive solution of the exponentially weighted regularized least-squares problem.
The derivation shows that the parameter matrix $\Theta_k$ that minimizes the cost function \eqref{eq:cost_def} admits both a closed-form expression and a recursive update that can be implemented online.

\subsection{Problem Formulation}

Consider the cost function
\begin{align}
    J(k,\Theta)
        &=
            \sum_{i=1}^k
            \lambda^{k-i}
            (y_i - \Theta x_i )^\rmT(y_i - \Theta x_i )
            + \tr \lambda^{k} R_\Theta \Theta \Theta^\rmT,
    \label{eq:cost_def}
\end{align}
where for all $i\geq 1,$ $y_i \in \BBR^{l_y}$ and
$x_i \in \BBR^{l_x},$ 
$\Theta \in \BBR^{l_y \times l_x}$, 
$R_\Theta \in \BBR^{l_y \times l_y}$ is positive definite, and 
$\lambda \in (0,1]$ is the forgetting factor.

The first term in \eqref{eq:cost_def} penalizes the exponentially weighted prediction error, where the forgetting factor $\lambda$ determines how quickly past data are forgotten.
The second term provides Tikhonov regularization and ensures that a unique minimizer exists. 

Note that the cost function can be written as
\begin{align}
    J(k,\Theta)
        &=
            \sum_{i=1}^k \lambda^{k-i} y_i^\rmT y_i 
            + \tr \SX_k \Theta ^\rmT \Theta   
            - 2 \tr \SY_k \Theta,
\end{align}
where
\begin{align}
    \SX_k 
        &\isdef 
            \sum_{i=1}^k  \lambda^{k-i} x_i x_i^\rmT + \lambda^{k} R_\Theta   
    , \\
    \SY_k 
        &\isdef 
            \sum_{i=1}^k  \lambda^{k-i} x_i y_i^\rmT .
\end{align}

The following identity from matrix differential calculus is used to compute the gradient of the cost function with respect to $\Theta$.
\begin{fact}
    \label{fact:trace_derivatives}
    Let $\Theta\in \BBR^{l_y \times l_x}$, $A \in \BBR^{l_y \times l_y},$ and $B\in \BBR^{l_x \times l_x}$ be matrices. Then,
    \begin{align}
        \frac{\partial}{\partial \Theta} \tr B \Theta^\rmT  \Theta
            &=
                \Theta (B+B^\rmT ) , 
        \\
        \frac{\partial}{\partial \Theta} \tr A \Theta B
            &=
                A^\rmT B. 
    \end{align}
\end{fact}
\begin{proof}
    See \cite{petersen2008matrix}.
\end{proof}

\subsection{Closed-Form Solution}
The following proposition provides a closed-form solution of the minimizer. 
\begin{proposition}
[\textbf{Closed-form minimizer of the weighted least-squares problem.}]
    \label{prop:theta_k_def}
    Consider the cost function \eqref{eq:cost_def}.
    Let $\Theta_k$ denote the minimizer of the cost function \eqref{eq:cost_def}.
    Then, 
    \begin{align}
        \Theta_k
            =
                \SY_k^\rmT \SX_k^{-1}.
    \end{align}
\end{proposition}
\begin{proof}
Using Fact \ref{fact:trace_derivatives}, it follows that 
\begin{align}
    \frac{\partial}{\partial \Theta} 
    \tr \SX_k \Theta ^\rmT \Theta 
        &=
            2 \Theta \SX_k , 
    \quad 
    \frac{\partial}{\partial \Theta} \tr\SY_k \Theta
        =
            \SY_k^\rmT, 
\end{align}
and thus
\begin{align}
     \frac{\partial}{\partial \Theta}
     J(k,\Theta)
        =
            2 \Theta \SX_k - 2\SY_k^\rmT.
\end{align}
Setting the gradient equal to zero yields the minimizer. 
\end{proof}

\subsection{Recursive Update}
The following proposition provides a recursive formula for computing the closed-form solution of the minimizer. 
This result provides the recursive update used in Step 3 of the DMAC algorithm described in Section \ref{sec:DMAC_algorithm}.
\begin{proposition}
    [\textbf{Recursive update for exponentially weighted least squares.}]
    \label{prop:theta_k_recursive}    
    Consider the cost function \eqref{eq:cost_def}.
    Let $\Theta_k$ denote the minimizer of the cost function \eqref{eq:cost_def}.
    Then, the minimizer $\Theta_k$ satisfies
    \begin{align}
        \Theta_k
            &=
                \Theta_{k-1} 
                +
                \left(
                    y_k - \Theta_{k-1} x_k
                \right)
                x_k^\rmT \SP_k  
            , \\
        \SP_k
            &=
                \lambda \inv \SP_{k-1} 
                -
                \lambda \inv \SP_{k-1} x_k 
                \gamma_k\inv
                x_k^\rmT \SP_{k-1},
    \end{align}
    where
    $\gamma_k \isdef \lambda  +  x_k^\rmT \SP_{k-1} x_k,$ and 
    $\SP_0 \isdef R_\Theta\inv. $
\end{proposition}

\begin{proof}
\hypertarget{proof:theta_k_rec}{Proof of Proposition \ref{prop:theta_k_recursive}.}
Note that 
\begin{align}
    \SX_k 
        &=
            \sum_{i=1}^k  \lambda^{k-i} x_i x_i^\rmT + \lambda^{k} R_\rma
        =
            \lambda \SX_{k-1} + x_k x_k^\rmT, 
    \\
    \SY_k 
        &=
            \sum_{i=1}^k  \lambda^{k-i} x_i y_i^\rmT 
        =
            \lambda \SY_{k-1} + x_k y_k^\rmT.
\end{align}

Define $\SP_k \isdef \SX_k^{-1} .$ Then, 
\begin{align}
    \SP_k
        &=
            \left( \lambda \SX_{k-1} + x_k x_k^\rmT \right)^{-1} 
        \nn \\ 
        &=
            \lambda \inv \SX_{k-1}\inv 
            -
            \lambda \inv \SX_{k-1}\inv 
            x_k 
            \gamma_k \inv 
            x_k^\rmT 
            \lambda \inv \SX_{k-1}\inv
        \nn \\ 
        &=
            \lambda \inv \SP_{k-1} 
            -
            \lambda \inv \SP_{k-1} x_k 
            \gamma_k \inv 
            x_k^\rmT \SP_{k-1},
        \\
    \Theta_k
        &=
            \SY_k^\rmT \SX_k\inv
        \nn \\
        &=
            (\lambda \SY_{k-1}^\rmT + y_k x_k^\rmT) \SX_k\inv
        \nn \\
        &=
            (\lambda \Theta_{k-1} \SX_{k-1} + y_k x_k^\rmT) \SX_k\inv
        \nn \\
        &=
            (\Theta_{k-1} (\SX_{k} - x_k x_k^\rmT) + y_k x_k^\rmT) \SX_k\inv
        \nn \\
        &=
            \Theta_{k-1} \SX_{k} \SX_k\inv 
            - \Theta_{k-1} x_k x_k^\rmT \SX_k\inv
            + y_k x_k^\rmT \SX_k\inv
        \nn \\
        &=
            \Theta_{k-1} 
            - \Theta_{k-1} x_k x_k^\rmT \SP_k
            + y_k x_k^\rmT \SP_k
        \nn \\
        &=
            \Theta_{k-1} 
            +
            \left(
                y_k - \Theta_{k-1} x_k
            \right)
            x_k^\rmT \SP_k            ,
\end{align}
which completes the proof. 
\end{proof}



\section{Full-State Feedback Control with Integral Action}
\label{sec:FSFi}
This section summarizes the full-state feedback controller with integral action used in the DMAC framework.
\subsection{Augmented System Formulation}
Consider the discrete-time linear system
\begin{align}
    x_{k+1} &= Ax_k + B u_k, \\
    y_k &= Cx_k,
\end{align}
where $x_k \in \BBR^{l_x}$ is the state, $u_k \in \BBR^{l_u}$ is the control input, and $y_k \in \BBR^{l_y}$ is the measured output.

To achieve reference tracking for a command signal $r_k$, define the integrator state
\begin{align}
    q_k \isdef \sum_{i=0}^k (r_i - y_i) \in \BBR^{l_y}.
\end{align}
Note that the integrator state satisfies
\begin{align}
    q_{k+1} = q_k + e_k,
    \label{eq:integrator_state}
\end{align}
where the output error is
\begin{align}
    e_k \isdef r_k - y_k = r_k - Cx_k.
\end{align}

Next, define the augmented state
\begin{align}
    x_{\rma,k} =
    \matl
        x_k \\
        q_k
    \matr
    \in \BBR^{l_x + l_y}.
\end{align}
The augmented system dynamics then are
\begin{align}
    x_{\rma,k+1} &= A_\rma x_{\rma,k} + B_\rma u_k + B_\rmr r_k, \\
    y_k &= C_\rma x_{\rma,k},
\end{align}
where
\begin{align}
A_\rma &\isdef
\matl
A & 0 \\
-C & I
\matr,
\quad
B_\rma \isdef
\matl
B \\
0
\matr,
\quad
B_\rmr \isdef
\matl
0 \\
I
\matr,
\\
C_\rma &\isdef
\matl
C & 0
\matr .
\end{align}

\subsection{Controller Gain Computation}
To stabilize the system and track the reference command $r_k$, consider the control law
\begin{align}
    u_k = K_x x_k + K_q q_k,
\end{align}
where $K_x \in \BBR^{l_u \times l_x}$ is the state-feedback gain matrix and $K_q \in \BBR^{l_u \times l_y}$ is the integral gain.
Defining
\begin{align}
    K_\rma \isdef \matl K_x & K_q \matr ,
\end{align}
the control law can be written as
\begin{align}
    u_k = K_\rma x_{\rma,k}.
\end{align}
Substituting the control law into the augmented dynamics yields
\begin{align}
    x_{\rma,k+1}
    =
    (A_\rma + B_\rma K_\rma)x_{\rma,k}
    +
    B_\rmr r_k .
\end{align}

Note that if the pair $(A_\rma,B_\rma)$ is stabilizable, then there exists a gain matrix $K_\rma$ such that $(A_\rma+B_\rma K_\rma)$ is Schur stable.
The gain matrix $K_\rma$ can be computed using standard state-feedback design techniques, such as pole placement or linear quadratic regulation.
The following proposition provides conditions under which the pair $(A_\rma,B_\rma)$ is stabilizable.
\begin{proposition}
[\textbf{Integral action and zeros at $z=1$}]
Assume $(A,B)$ is stabilizable. Then $(A_\rma,B_\rma)$ is stabilizable if and only if the triple $(A,B,C)$ has no invariant zeros at $z=1$, or equivalently,
\begin{align}
{\rm rank }
\begin{bmatrix}
A-I & B\\
C & 0
\end{bmatrix}
= l_x + l_y .
\end{align}
\end{proposition}
\begin{proof}
    See \cite{rugh1996linear}.
\end{proof}

Finally, if $r_k = r$ is constant, then $x_{\rma,k}$ converges to a constant equilibrium point. 
Consequently, the integrator state $q_k$ converges, which implies that the tracking error $e_k$ converges to zero and therefore the output $y_k$ converges to $r$.


\section{Proofs}
\label{appndx:proof}
\subsection{Proof of Proposition \ref{prop:P_pos_def}}
\label{proof:P_pos_def}
\begin{proof}
    It follows from the Sherman--Morrison--Woodbury formula, \cite{deng2011generalization}, that 
    \begin{align}
        (\lambda \SP_k^{-1} + &\phi_k \phi_k^\rmT)^{-1}
            =
            \lambda^{-1}\SP_k
            \nn \\
            &-\lambda^{-1}\SP_k\phi_k\big(\lambda+\phi_k^\rmT\SP_k\phi_k\big)^{-1}\phi_k^\rmT\SP_k
            .
            \nn
    \end{align}
    Next, note that it follows from \eqref{eq:SP_k} that the RSH of the equation above is equal to $\SP_{k+1},$ which proves \eqref{eq:SPinv_rec}.

    Next, note that, for $k=0$, $\SP_1\inv$ is a sum of two positive definite matrices and thus is positive definite. 
    It follows by induction that, for all $k\ge0,$ $\SP_k\inv$ is positive definite, which implies that $\SP_k$ is positive definite.

    Since $\phi_k$ is persistently exciting, $\alpha, \beta,$ and $N$ exist such that \eqref{eq:Pk_bounds} holds. 
    Since $\beta$ is finite, it follows that there exists $\bar \phi \in \BBR$ such that, for all $k \ge 0$, $\|\phi_k\| \le \bar{\phi}.$ 

    Next, it follows from \eqref{eq:SPinv_rec} that
    \begin{align}
        \SP_{k+1}^{-1}
        &=
        \lambda^{k+1} \SP_0^{-1}
        +
        \sum_{i=0}^{k} \lambda^{k-i} \phi_i \phi_i^\rmT 
        \nn \\
        &\le
        \lambda^{k+1} \SP_0^{-1}
        +
        \sum_{i=0}^{k} \lambda^{k-i} \bar{\phi}^2 I_p 
        \nn \\
        &\le
        \SP_0^{-1}
        +
        \frac{\bar{\phi}^2}{1-\lambda} I_p.
        \nn
    \end{align}
    First, consider the case where $\lambda \in (0,1).$
    In this case, $\SP_{k+1} \inv$ is bounded from below and consequently $\SP_{k+1}$ is bounded from above. 
    
    On the other hand, in the case where $\lambda=1,$ $\SP_{k}$ satisfies
    \begin{align}
        \SP_{k+1}^{-1}
            =
            \SP_0^{-1}
            +
            \sum_{i=0}^{k} \phi_i \phi_i^\rmT 
            \geq
            \SP_0^{-1}
            +
            \alpha \lfloor k/N \rfloor I_p,
        \nn 
    \end{align}
    where $\lfloor k/N \rfloor$ is the smallest integer larger than $k/N.$
    Conseqeuntly, $\SP_{k+1}\inv$ grows without bounds and thus $\SP_{k+1} \to 0$ as $k \to \infty.$  
    \end{proof}

\subsection{Proof of Proposition \ref{prop:Theta_tilde_update}}
\label{proof:Theta_tilde_update}
\begin{proof}
Using \eqref{eq:dmac_local_model} and \eqref{eq:Theta_k}, it follows that 
\begin{align}
    \widetilde\Theta_{k+1} \nonumber
    &= \Theta_{k+1} - \Theta^\star \\ \nonumber
    &= \Theta_k-\Theta^\star
        +\big(\xi_{k+1}-\Theta_k\phi_k\big)\phi_k^\rmT \SP_{k+1} \\ \nonumber
    &= \widetilde\Theta_k
        +\big(\Theta^\star\phi_k - (\widetilde\Theta_k+\Theta^\star)\phi_k\big)\phi_k^\rmT \SP_{k+1} \\ \nonumber
    &= \widetilde\Theta_k
        -\widetilde\Theta_k\phi_k\phi_k^\rmT \SP_{k+1} \\
    &= \widetilde\Theta_k\big(I - \phi_k\phi_k^\rmT \SP_{k+1}\big). \nn
\end{align}
Finally, substituting $\phi_k \phi_k^\rmT$ from \eqref{eq:SPinv_rec} in the equation above yields the second equality. 
\end{proof}

\subsection{Proof of Theorem \ref{thm:matrix_rls_stability}}.
\label{proof:matrix_rls_stability}
\begin{proof}
Consider the function 
\begin{align}
    V_k 
        \isdef
            \tr\big(
            \widetilde{\Theta}_k \SP_k^{-1} \widetilde{\Theta}_k^\rmT
            \big). 
    \nn 
\end{align}
Note that, for all $k \ge 0$ and $\widetilde{\Theta}_k \neq 0,$ $V_k > 0$ since  $\SP_k$ is positive definite.
Furthermore, $V_k = 0$ if and only if $\widetilde{\Theta}_k = 0.$


Using \eqref{eq:est_error_update_v2}, it follows that 
\begin{align}
    V_{k+1}
        &=
            \tr
            \big(
            \widetilde{\Theta}_{k+1} \SP_{k+1}^{-1} \widetilde{\Theta}_{k+1}^\rmT
            \big)
        \nn \\
        &=
            \lambda^2
            \tr
            \big(
            \widetilde{\Theta}_k
            \SP_k^{-1} \SP_{k+1}
            \SP_{k+1}^{-1}
            \SP_{k+1} \SP_k^{-1}
            \widetilde{\Theta}_k^\rmT 
            \big)
            \nn \\
        &=
            \lambda^2
            \tr\big(
                \widetilde{\Theta}_k
                \SP_k^{-1} \SP_{k+1} \SP_k^{-1}
                \widetilde{\Theta}_k^\rmT
            \big).
            \label{eq:Vkplus1_expanded_1}
\end{align}

Post-multiplying \eqref{eq:SPinv_rec} by $\SP_{k+1}$ yields
gives
\begin{align}
    \SP_k^{-1} \SP_{k+1}
    =
    \lambda^{-1}
    \Big(
        I - \phi_k \phi_k^\rmT \SP_{k+1}
    \Big).
    \label{eq:Sk_inv_Sk1}
\end{align}
Next, substituting \eqref{eq:Sk_inv_Sk1} into \eqref{eq:Vkplus1_expanded_1} yields
\begin{align}
    V_{k+1}
    &=
    \lambda
    \tr\big(
        \widetilde{\Theta}_k
        \Big(
            I - \phi_k \phi_k^\rmT \SP_{k+1}
        \Big)
        \SP_k^{-1}
        \widetilde{\Theta}_k^\rmT
    \big)
    \nn
    \\
    &=
    \lambda
    \tr\big(
        \widetilde{\Theta}_k \SP_k^{-1} \widetilde{\Theta}_k^\rmT
    \big)
    -
    \lambda
    \tr\big(
        \widetilde{\Theta}_k
        \phi_k \phi_k^\rmT
        \SP_{k+1} \SP_k^{-1}
        \widetilde{\Theta}_k^\rmT
    \big), 
    \nn
\end{align}
which implies that 
\begin{align}
\label{eq:Vk_diff_final}
    V_{k+1} - V_k
    &=
    -(1-\lambda) V_k
    -
    \lambda
    \tr\big(
        \widetilde{\Theta}_k
        \phi_k \phi_k^\rmT
        \SP_{k+1} \SP_k^{-1}
        \widetilde{\Theta}_k^\rmT
    \big)
    \nn \\
    &\le
        -(1-\lambda) V_k
    \nn \\
    &\le
        0,
\end{align}
since 
\begin{align}
    \tr\big(
        \widetilde{\Theta}_k
        \phi_k \phi_k^\rmT
        \SP_{k+1} \SP_k^{-1}
        \widetilde{\Theta}_k^\rmT
    \big)
    =
    \tr\big(
        \widetilde{\Theta}_k^\rmT
        \widetilde{\Theta}_k
        \phi_k \phi_k^\rmT
        \SP_{k+1} \SP_k^{-1}        
    \big)
    \ge 0.
    \nn
\end{align}



Note that if $\lambda = 1$, then, for all $k \geq 0,$
\begin{align}
    V_{k+1} - V_k \leq 0, \nn
\end{align}
which implies Lyapunov stability of the equilibrium $\widetilde \Theta_k = 0.$
Furthermore, if $\lambda \in (0,1)$, then, for all $k \geq 0,$
\begin{align}
    V_{k+1} - V_k < 0, \nn
\end{align}
which implies globally geometric stability of the equilibrium $\widetilde \Theta_k = 0.$
\end{proof}

\subsection{Proof of Theorem \ref{thm:closed_loop_stability_new}}
\label{proof:closed_loop_stability_new}
\begin{proof}
Define the composite Lyapunov function
\begin{align}
    \SV_k(\xi_k, \widetilde{\Theta}_k)
        \isdef
            \SSS_k + V_k, \nn
\end{align}
where
\begin{align}
     \SSS_k
        &\isdef
            \xi_k^\rmT P \xi_k,
     \quad
     V_k
        \isdef
            \tr\big(
            \widetilde{\Theta}_k \SP_k^{-1} \widetilde{\Theta}_k^\rmT
            \big), \nn
\end{align}
and $\widetilde{\Theta}_k$ is given by \eqref{eq:estimation_error}.
Furthermore, $\SP_k$ given by \eqref{eq:SP_k} is positive definite, and $P$ is the positive definite solution of the discrete-time Lyapunov equation
\begin{align}
    (A_c^\star)^\rmT P A_c^\star - P = -Q,
    \nn
\end{align}
where
\begin{align}
    A_c^\star \isdef A^\star + B^\star K^\star.
    \nn
\end{align}
Thus, $\SV_k(\xi_k, \widetilde{\Theta}_k) \ge 0$ since $P$ and $\SP_k$ are positive definite.

Since $v_k$ is i.i.d.\ and zero-mean, it follows from Proposition \ref{prop:exciting_input} that the regressor $\phi_k$ given by \eqref{eq:regressor_def} is persistently exciting. Consequently, by Theorem \ref{thm:matrix_rls_stability},
\begin{align}
    V_{k+1}-V_k \le 0. \nn
\end{align}

Next, define $\Delta \SSS_k \isdef \SSS_{k+1}- \SSS_k.$ Then, 
\begin{align}
    \Delta \SSS_k
        =
            -\xi_k^\rmT Q \xi_k
            + T_{1,k} + T_{2,k} + T_{3,k} + T_{4,k},
    \label{eq:state_diff_full}
\end{align}
where
\begin{align}
    T_{1,k}
    &\isdef
    2 \xi_k^\rmT (A_c^\star)^\rmT P B^\star \widetilde K_k \xi_k,
    \nn
    \\
    T_{2,k}
    &\isdef
    \xi_k^\rmT \widetilde K_k^\rmT B^{\star\rmT} P B^\star \widetilde K_k \xi_k,
    \nn
    \\
    T_{3,k}
    &\isdef
    2 \xi_k^\rmT (A_c^\star)^\rmT P B^\star v_k
    + v_k^\rmT B^{\star\rmT} P B^\star v_k,
    \nn
    \\
    T_{4,k}
    &\isdef
    2 \xi_k^\rmT \widetilde K_k^\rmT B^{\star\rmT} P B^\star v_k,
    \nn
\end{align}
with $\widetilde K_k \isdef K_k - K^\star$.
By Proposition \ref{prop:lqr_gain_convergence_bound}, there exist constants $c_K>0$ and $k_0\in\mathbb{N}$ such that, for all $k\ge k_0$,
\begin{align}
        \|\widetilde K_k\|
        =
        \|K_k-K^\star\|
        \le
        c_K \|\widetilde\Theta_k\|.
        \nn
    \end{align}

\textbf{Bounding the perturbation terms.}
Next, we bound each term. First,
\begin{align}
    |T_{1,k}
    &=
    2\big|\xi_k^\rmT (A_c^\star)^\rmT P B^\star \widetilde K_k \xi_k\big|
    \nn \\
    &\le
    2 \|\xi_k\| \,\|(A_c^\star)^\rmT P B^\star\|\,\|\widetilde K_k\|\,\|\xi_k\|
    \nn \\
    &\le
    2 \|A_c^\star\|\,\|P\|\,\|B^\star\|\,\|\widetilde K_k\|\,\|\xi_k\|^2
    \nn \\
    &\le
    c_1 \|\xi_k\|^2 \|\widetilde\Theta_k\|,
    \label{eq:T1_bound_final}
\end{align}
where
\begin{align}
    c_1 \isdef 2 \|A_c^\star\|\,\|P\|\,\|B^\star\|\,c_K.
    \nn
\end{align}
Next,
\begin{align}
    |T_{2,k}
    &=
    \big|
    \xi_k^\rmT \widetilde K_k^\rmT B^{\star\rmT} P B^\star \widetilde K_k \xi_k
    \big|
    \nn \\
    &\le
    \|B^\star\|^2 \|P\|\,\|\widetilde K_k\|^2\,\|\xi_k\|^2
    \nn \\
    &\le
    c_2 \|\widetilde\Theta_k\|^2\,\|\xi_k\|^2,
    \label{eq:T2_bound_final}
\end{align}
where
    $c_2 \isdef \|B^\star\|^2 \|P\|\,c_K^2.$

Next, for $\epsilon>0,$ note that
\begin{align}
    T_{3,k}
    &\le
    2\|\xi_k\|\,\|(A_c^\star)^\rmT P B^\star\|\,\|v_k\|
    +
    \|B^{\star\rmT} P B^\star\|\,\|v_k\|^2
    \nn \\
    &\le
    \epsilon \|\xi_k\|^2
    +
    \left(
        \frac{1}{\epsilon}\|(A_c^\star)^\rmT P B^\star\|^2
        + \|B^\star\|^2 \|P\|
    \right)\|v_k\|^2
    \nn \\
    &\le
    c_{3x}\|\xi_k\|^2 + c_{3v}\|v_k\|^2,
    \label{eq:T3_bound_final}
\end{align}
where
\begin{align}
    c_{3x} \isdef \epsilon,
    \qquad
    c_{3v} \isdef
    \frac{1}{\epsilon}\|(A_c^\star)^\rmT P B^\star\|^2
    + \|B^\star\|^2 \|P\|.
    \nn
\end{align}

Next, note that
\begin{align}
    |T_{4,k}
    &\le
    2 \|\xi_k\|\,\|\widetilde K_k\|\,\|B^\star\|^2\|P\|\,\|v_k\|
    \nn \\
    &\le
    2 c_K \|B^\star\|^2\|P\|\,
    \|\xi_k\|\,\|\widetilde\Theta_k\|\,\|v_k\|.
    \nn
\end{align}
Applying Young's inequality (Lemma \ref{lem:young}), for $\eta>0,$ yields
\begin{align}
    |T_{4,k}
    &\le
    \eta \|\xi_k\|^2\|\widetilde\Theta_k\|^2
    +
    \frac{1}{\eta}
    \big(c_K\|B^\star\|^2\|P\|\big)^2
    \|v_k\|^2
    \nn \\
    &\le
    c_{4\theta}\|\xi_k\|^2\|\widetilde\Theta_k\|^2 + c_{4v}\|v_k\|^2,
    \label{eq:T4_bound_final}
\end{align}
where
    $c_{4\theta} \isdef \eta,$
    and
    $c_{4v} \isdef \frac{1}{\eta}\big(c_K\|B^\star\|^2\|P\|\big)^2.$

\textbf{Bounding the Lyapunov function.}
Substituting \eqref{eq:T1_bound_final}, \eqref{eq:T2_bound_final},
\eqref{eq:T3_bound_final}, and \eqref{eq:T4_bound_final} into
\eqref{eq:state_diff_full} yields
\begin{align}
    \Delta \SSS_k
    &\le
    -\xi_k^\rmT Q \xi_k
    + c_1 \|\xi_k\|^2 \|\widetilde\Theta_k\|
    + c_2 \|\xi_k\|^2 \|\widetilde\Theta_k\|^2
    \nn \\
    &\quad
    + c_{3x}\|\xi_k\|^2
    + c_{3v}\|v_k\|^2
    + c_{4\theta}\|\xi_k\|^2\|\widetilde\Theta_k\|^2
    \nn \\
    &\quad
    + c_{4v}\|v_k\|^2.
    \label{eq:state_bound_collected_raw}
\end{align}
Combining like terms and using $\xi_k^\rmT Q \xi_k \ge \lambda_{\min}(Q)\|\xi_k\|^2,$ it follows that
\begin{align}
    \Delta \SSS_k
    &\le
    -\lambda_{\min}(Q)\|\xi_k\|^2
    + c_1 \|\xi_k\|^2 \|\widetilde\Theta_k\|
    + c_\theta \|\xi_k\|^2 \|\widetilde\Theta_k\|^2
    \nn \\
    &\quad
    + c_x \|\xi_k\|^2
    + c_v \|v_k\|^2,
    \label{eq:state_bound_collected}
\end{align}
where $c_\theta \isdef c_2 + c_{4\theta}$, $c_x \isdef c_{3x}$, and $c_v \isdef c_{3v}+c_{4v}$.
Next, note that $v_k$ is bounded by $\bar v$ and assuming $\lambda_{\min}(Q) > c_x,$ define $c_q \isdef \lambda_{\min}(Q) - c_x > 0,$ it follows that 
\begin{align}
    \Delta \SSS_k
    &\le
    -c_q\,\|\xi_k\|^2
    + c_1 \|\xi_k\|^2 \|\widetilde\Theta_k\|
    + c_\theta \|\xi_k\|^2 \|\widetilde\Theta_k\|^2
    + c_v \bar v^2.
    \label{eq:state_bound_with_vbar}
\end{align}

\textbf{Small-Gain Condition from Parameter Convergence.}
Since the matrix RLS estimator converges globally and geometrically under persistent excitation for $0<\lambda<1$ (Theorem \ref{thm:matrix_rls_stability}), $\|\widetilde\Theta_k\|\to 0$ as $k\to\infty$.
Fix any $\delta\in(0,1)$ and choose $k$ sufficiently large such that
\begin{align}
    c_1 \|\widetilde\Theta_k\| + c_\theta \|\widetilde\Theta_k\|^2
    \le
    \delta\,c_q.
    \label{eq:theta_small_condition}
\end{align}
Then \eqref{eq:state_bound_with_vbar} implies, for all such $k$,
\begin{align}
    \Delta \SSS_k
    \le
    -(1-\delta)c_q\,\|\xi_k\|^2
    + c_v \bar v^2.
    \label{eq:state_ultimate_decrease}
\end{align}

\textbf{Ultimate Boundedness Radius.}
Define
\begin{align}
    r_v^2
    \isdef
    \frac{c_v}{(1-\delta)c_q}\,\bar v^2.
    \label{eq:rv_def}
\end{align}
Then \eqref{eq:state_ultimate_decrease} implies that whenever $\|\xi_k\|^2 > r_v^2$,
    $\Delta \SSS_k < 0.$
Suppose that for some $k \ge k_0$,
    $\|\xi_k\| > r_v.$
Then, it follows that
$\SSS_{k+1} < \SSS_k.$
Hence, whenever $\|\xi_k\| > r_v$, the sequence $\{\SSS_k\}$ is strictly decreasing.
Since $\SSS_k \ge 0$, it follows that $\|\xi_k\| > r_v$ cannot hold for all $k \ge k_0$.
Therefore, there exists a finite index $k_1 \ge k_0$ such that
$\|\xi_{k_1}\| \le r_v .$

\textbf{Norm Equivalence and Explicit State Bound.}
Since $P=P^\rmT>0$, $\SSS_k$ is norm-equivalent to $\|\xi_k\|^2$, that is, for all $k \geq 0,$
\begin{align}
    \lambda_{\min}(P)\|\xi_k\|^2
    \le
    \SSS_k
    \le
    \lambda_{\max}(P)\|\xi_k\|^2.
    \label{eq:norm_equiv_P}
\end{align}
Thus, $\|\xi_k\|\le r_v$ implies $\SSS_k \le \lambda_{\max}(P)\,r_v^2$.
Hence, for all $k\ge k_1$,
\begin{align}
    \|\xi_k\|^2
    \le
    \frac{\lambda_{\max}(P)}{\lambda_{\min}(P)}\,r_v^2,
\end{align}
From \eqref{eq:rv_def},
\begin{align}
    r_v
    =
    \sqrt{\frac{c_v}{(1-\delta)c_q}}\,\bar v .
\end{align}
Equivalently, $\xi_k$ is uniformly ultimately bounded and enters and remains in the set
\begin{align}
    \left\{\xi:\ \|\xi\|\le
    \underbrace{
    \sqrt{\frac{\lambda_{\max}(P)}{\lambda_{\min}(P)}}
    \sqrt{\frac{c_v}{(1-\delta)c_q}}
    }_{\displaystyle c}
    \,\bar v .\right\},
\end{align}
whose radius scales linearly with $\bar v$. Therefore,
\begin{align}
    \limsup_{k\to\infty}\|\xi_k\|
    \le
    c\,\bar v.
\end{align}

\textbf{Closed-Loop Stability.}
Since $V_{k+1}-V_k\le 0$ and $\SSS_k$ is uniformly ultimately bounded, it follows that
$\SV_k(\xi_k,\widetilde\Theta_k)=\SSS_k+V_k$ is bounded and nonnegative.
Moreover, $\widetilde\Theta_k$ is Lyapunov stable (and globally geometrically convergent for $0<\lambda<1$) by Theorem \ref{thm:matrix_rls_stability}.
Therefore, the closed-loop system is Lyapunov stable and the state $\xi_k$ is uniformly ultimately bounded.
%
\end{proof}

\begin{lemma}[Young's Inequality with Parameter]
\label{lem:young}
For any $a,b \in \BBR$ and any $\epsilon > 0,$ the following inequality holds.
\begin{align}
    2ab \le \epsilon a^2 + \frac{1}{\epsilon} b^2.
    \label{eq:YoungsIneq}
\end{align}
\end{lemma}

\begin{proof}
Since $(\sqrt{\epsilon}\,a - \tfrac{1}{\sqrt{\epsilon}}\,b)^2 \ge 0$ for all $a,b \in \BBR$ and $\epsilon>0,$ it follows that
\begin{align}
    0 
    &\le 
    \epsilon a^2 - 2ab + \frac{1}{\epsilon} b^2,
\end{align}
which, upon rearranging, yields \eqref{eq:YoungsIneq}.
\end{proof}



\section{Supporting Convergence Results}
This Appendix provides additional technical derivations that support the analysis presented in the main text. 
These derivations expand on intermediate steps that were omitted in the paper to improve readability and are included here for completeness.

\begin{definition}
[\textbf{Expected persistence of excitation.}]
\label{def:pe_expectation}
Consider, for all $k \ge 0$, a regressor $\phi_k \in \BBR^p$.
The regressor $\{\phi_k\}$ is called \emph{persistently exciting in expectation} if there exist
positive constants $\alpha$, $\beta$ and an integer $N>p$ such that, for all $k \ge 0$,
\begin{align}\label{eq:pe_expectation}
    \alpha I_p 
        \leq
            \BBE\left[\sum_{i=k}^{k+N} \phi_i \phi_i^\rmT\right]
        \leq
            \beta I_p .
\end{align}
\end{definition}

\begin{remark}
Definition \ref{def:pe_expectation} is a mean-square (expected) variant of the classical definition of persistence of excitation.
This notion is convenient in stochastic settings, that is, in the case where $\phi_k$ depends on random inputs,
and is sufficient for the covariance boundedness arguments used in the subsequent result.
\end{remark}

\begin{proposition}
    [\textbf{Persistence of excitation under random inputs.}]
    \label{prop:exciting_input}    
    Consider the system
    \begin{align}
        x_{k+1} = A x_k + B u_k,
    \end{align}
    where $x_k \in \BBR^n$ and $u_k \in \BBR^m.$
    Assume that $A \in \BBR^{n \times n}$ is Schur stable and the pair $(A,B)$ is controllable. 
    Assume that 
    $x_0$ is bounded and 
    $u_k$ is i.i.d., zero-mean with covariance $R.$
    Then, the regressor
    \begin{align}
        \phi_k
            \isdef 
                \matl 
                    x_k \\
                    u_k
                \matr
                \in 
                \BBR^{n+m}
    \end{align}
    is persistently exciting. 
\end{proposition}

\begin{proof}

For all $k \geq 0,$ define
\begin{align}
    S_k 
        \isdef
            \sum_{i=k}^{k+N} \phi_i \phi_i^\rmT
        =
            \sum_{i=k}^{k+N}
            \matl
                x_i x_i^\rmT & x_i u_i^\rmT \\
                u_i x_i^\rmT & u_i u_i^\rmT
            \matr.
    \label{eq:Sk_def}
\end{align}

\smallskip
\noindent\textbf{Lower bound.}
Since $\BBE[u_i u_i^\rmT]=R> 0$, it follows that
\begin{align}\label{eq:u_lower}
\BBE \left[\sum_{i=k}^{k+N} u_i u_i^\rmT\right]
=
(N+1)R
\geq
(N+1)\lambda_{\min}(R) I_m.
\end{align}
Next, note that, for $\ell\ge 1,$ $x_{k+\ell}$ can be written as
\begin{align}
    x_{k+\ell}
        =
            A^\ell x_k + \sum_{j=0}^{\ell-1} A^{\ell-1-j} B u_{k+j}.
    \nn
\end{align}
Since $(A,B)$ is controllable, there exists an integer $L\ge 1$ such that
\begin{align}
    \SC_L 
        \isdef
            \matl B & AB & \cdots & A^{L-1}B \matr
            \in \BBR^{n\times (Lm)}
    \nn
\end{align}
has full row rank.
Define
\begin{align}
    \tilde{x}_{k+L}
        \isdef
            \sum_{j=0}^{L-1} A^{L-1-j}B u_{k+j}
        =
            \SC_L \matl u_{k+L-1}\\ \vdots \\ u_k \matr.
    \nn
\end{align}
Since $\{u_k\}$ is i.i.d.\ with covariance $R$ and $\SC_L$ has full row rank, it follows that
\begin{align} 
    \BBE[\tilde{x}_{k+L}\tilde{x}_{k+L}^\rmT]
        =
            \SC_L (I_L\otimes R)\SC_L^\rmT.
    \nn
\end{align}
is positive definite. 
Hence, there exists $\eta>0$ such that, for all $k\ge 0,$
\begin{align}\label{eq:x_lower_one}
    \BBE[\tilde{x}_{k+L}\tilde{x}_{k+L}^\rmT] \geq \eta I_n
    .
\end{align}

Next, note that, for any $N\ge L,$
\begin{align}
    \sum_{i=k}^{k+N} x_i x_i^\rmT \geq x_{k+L}x_{k+L}^\rmT.
    \nn
\end{align}

Since $u_{k+L}$ is independent of $\{u_k,\ldots,u_{k+L-1}\}$ and satisfies $\BBE[u_{k+L}]=0$, it follows that
\begin{align}
    \BBE[x_{k+L} & u_{k+L}^\rmT]
        =
            \BBE \left[
                \left(A^L x_k + \sum_{j=0}^{L-1} A^{L-1-j} B u_{k+j}\right) u_{k+L}^\rmT
            \right]
            \nn \\
        &=
            A^L \BBE[x_k u_{k+L}^\rmT]
            +
            \sum_{j=0}^{L-1} A^{L-1-j} B \, \BBE[u_{k+j} u_{k+L}^\rmT] 
            \nn \\
        &=
            0.
        \nn 
\end{align}


Therefore, taking expectations in \eqref{eq:Sk_def} and using \eqref{eq:u_lower} and
\eqref{eq:x_lower_one}, yields, for any $N\ge L$,
\begin{align}\label{eq:Sk_lower}
    \BBE[S_k]
        =
        \matl
        \BBE \left[\sum_{i=k}^{k+N} x_i x_i^\rmT\right] & 0\\
        0 & \BBE \left[\sum_{i=k}^{k+N} u_i u_i^\rmT\right]
        \matr
        \geq
        \alpha I_{n+m},
\end{align}
where
\begin{align}
    \alpha \isdef \min\left\{\eta,\ (N+1)\lambda_{\min}(R)\right\} 
        >
            0.
    \nn
\end{align}

\smallskip
\noindent\textbf{Upper bound.}
Since $A$ is Schur and $x_0$ is bounded, 
there exists $c_x>0$ such that $\sup_{k\ge 0}\BBE[\|x_k\|^2]\le c_x$.
It follows that, for all $k\ge 0$,
\begin{align}
    \BBE \left[\sum_{i=k}^{k+N} x_i x_i^\rmT\right]
        \leq
            \sum_{i=k}^{k+N} \BBE[\|x_i\|^2]\, I_n
        \leq
            (N+1)c_x\, I_n. 
    \nn
\end{align}
Furthermore, since $\BBE[u_i u_i^\rmT]=R$, it follows that
\begin{align} 
    \BBE \left[\sum_{i=k}^{k+N} u_i u_i^\rmT\right]
    =
        (N+1)R
    \leq
        (N+1)\lambda_{\max}(R) I_m.
        \nn
\end{align}
Hence, for all $k\ge 0$, 
\begin{align}\label{eq:Sk_upper}
    \BBE[S_k]
    \leq
    \beta I_{n+m},
\end{align}
where
\begin{align}
    \beta 
        \isdef
            \max\left\{(N+1)c_x,\ (N+1)\lambda_{\max}(R)\right\} < \infty.
    \nn
\end{align}

Finally, \eqref{eq:Sk_lower} and \eqref{eq:Sk_upper} imply that there exist $\alpha>0$, $\beta>0$,
and an integer $N>n+m$ such that, for all $k\ge 0$,
\begin{align}
    \alpha I_{n+m}
        \leq
            \BBE \left[\sum_{i=k}^{k+N} \phi_i \phi_i^\rmT\right]
        \leq
            \beta I_{n+m}.
\nn
\end{align}
It thus follows from Definition \ref{def:pe_expectation} that $\phi_k$ is persistently exciting in expectation.
\end{proof}

The next result follows from continuity of the stabilizing solution of the discrete-time algebraic Riccati equation with respect to $(A,B)$ and continuity of the LQR gain map.

\begin{proposition}
[\textbf{Continuity and local Lipschitz property of the LQR gain.}]
\label{prop:lqr_gain_convergence_bound}
    Let $Q \in \BBR^{n \times n}$ and $R \in \BBR^{m \times m}$ be positive definite.     
    Consider matrices $A \in \BBR^{n \times n}$ and $B \in \BBR^{n \times m}$ such that 
    $(A,B)$ is stabilizable and $(A,Q^{1/2})$ is detectable. 
    Let $K$ denote the infinite-horizon LQR gain associated with the stabilizing solution of the discrete-time algebraic Riccati equation for $(A,B, Q, R)$.
    Assume that the sequence $(A_k, B_k)$ converges to $(A,B),$ where $(A_k,B_k)$ is stabilizable for all $k.$
    Then, the infinite-horizon LQR gain $K_k$ computed using $(A_k, B_k, Q, R)$ converges to $K.$ 

    Furthermore, define $\Theta \isdef \matl A & B \matr$ and $\Theta_k \isdef \matl A_k & B_k \matr$.
    There exist constants $\rho>0$ and $c_K>0$ such that if $\|\Theta_k-\Theta\|\le \rho$, then
    \begin{align}
        \|K_k - K\|
        \le
        c_K \|\Theta_k - \Theta\|.
        \label{eq:lqr_gain_lipschitz}
    \end{align}


\end{proposition}

\begin{proof}
For each $k$, let $P_k$ denote the stabilizing symmetric positive definite solution of the discrete-time algebraic Riccati equation (DARE) associated with $(A_k,B_k,Q,R)$, and let $P$ denote the stabilizing solution of the DARE associated with $(A,B,Q,R)$. Define
\begin{align}
    S_k &\isdef R + B_k^\rmT P_k B_k, 
    \nn
    \\
    S &\isdef R + B^\rmT P B.
    \nn
\end{align}
The corresponding infinite-horizon LQR gains are
\begin{align}
    K_k &= S_k^{-1}B_k^\rmT P_k A_k, 
    \nn
    \\    
    K &= S^{-1}B^\rmT P A.    
    \nn
\end{align}

Since $(A_k,B_k)\to(A,B)$, and the pairs $(A_k,B_k)$ remain stabilizable and $(A_k,Q^{1/2})$ remain detectable, and hence the DARE for $(A_k,B_k,Q,R)$ admits a unique positive definite solution $P_k> 0.$

Next, note that for the stabilizing DARE solution that, on the set of stabilizable/detectable quadruples $(A,B,Q,R)$ with $Q> 0$, $R> 0$ and stabilizing solution, the mapping
\begin{align}
(A,B)\mapsto P(A,B)
\nn
\end{align}
is continuous \cite{lancaster1995algebraic,aboukandil2003matrix}.
Therefore, as $k\to\infty,$ $P_k \to P.$

We now show that this implies $K_k\to K$. 

Next, since $P_k\to P$ and $(A_k,B_k)\to(A,B),$ it follows that 
\begin{align}
B_k^\rmT P_k B_k \to B^\rmT P B, 
\nn
\\
B_k^\rmT P_k A_k \to B^\rmT P A,
\nn
\end{align}
which implies that $S_k \to S$.
Since $R> 0$ and $P_k> 0$, it follows that, for all $k\geq0,$ $S_k \geq R > 0,$ and thus $S_k$ is invertible, and inversion is continuous on the cone of symmetric positive definite matrices, which implies that $S_k^{-1}\to S^{-1}$.

Next, note that 
\begin{align}
K_k - K
    &= 
        S_k^{-1}B_k^\rmT P_k A_k - S^{-1}B^\rmT P A
    \nn \\
    &=
        (S_k^{-1}-S^{-1})B_k^\rmT P_k A_k + S^{-1}\bigl(B_k^\rmT P_k A_k - B^\rmT P A\bigr),
    \label{eq:Kk_minus_K}
\end{align}
which implies that 
\begin{align}
\|K_k-K\|
    &\le
        \|S_k^{-1}-S^{-1}\|\,\|B_k^\rmT P_k A_k\| 
    \nn \\
    &+ \|S^{-1}\|\,\|B_k^\rmT P_k A_k - B^\rmT P A\|.
    \nn
\end{align}
Since $S_k^{-1}\to S^{-1}$ and $B_k^\rmT P_k A_k \to B^\rmT P A,$ it follows that the right-hand side converges to $0,$ and thus $\|K_k-K\|\to 0$.

Finally, note that since the mapping $(A,B)\mapsto P(A,B)$ is locally Lipschitz in a neighborhood of $(A,B)$ on the set of stabilizable/detectable pairs with $Q>0$ and $R>0$ \cite{lancaster1995algebraic,aboukandil2003matrix},
there exist constants $\rho>0$ and $c_P>0$ such that, for all $\Theta_k$ satisfying $\|\Theta_k-\Theta\|\le \rho$ and remaining stabilizable/detectable,
\begin{align}
    \|P_k - P\|
    \le
    c_P \|\Theta_k - \Theta\|.
\end{align}
Moreover, for $\|\Theta_k-\Theta\|\le \rho$ with $\rho$ sufficiently small, the matrices $A_k,B_k,P_k$ remain uniformly bounded and $S_k$ remains uniformly positive definite, which implies that the matrix inversion is locally Lipschitz on $\{S_k\}$. 
Consequently, there exists a constant $c_{\rm inv}>0$ such that
\begin{align}
\|S_k^{-1}-S^{-1}\|
\le
c_{\rm inv}\|S_k-S\|.
\nn
\end{align}
Similarly, there exist constants $c_1,c_2>0$ such that
\begin{align}
\|S_k-S\|
&\le
c_1 \|\Theta_k-\Theta\|,
\nn\\
\|B_k^\rmT P_k A_k - B^\rmT P A\|
&\le
c_2 \|\Theta_k-\Theta\|.
\nn
\end{align}
Combining these bounds yields
\begin{align}
    \|K_k-K\|
    &\le
        \|S_k^{-1}-S^{-1}\|\,\|B_k^\rmT P_k A_k\|
        \nn \\
        &+
        \|S^{-1}\|\,\|B_k^\rmT P_k A_k - B^\rmT P A\|
        \nn\\
        &\le
        c_K \|\Theta_k-\Theta\|,
        \nn
\end{align}
for some constant $c_K>0.$
\end{proof}

\section*{Funding Data}

This research was supported by the Office of Naval Research grant N00014-23-1-2468.












\bibliographystyle{asmejour} 
\bibliography{DMDbib}     
\end{document}